\documentclass[aoas,preprint]{imsart}

\usepackage{amsthm,amsmath,amssymb}
\usepackage{natbib}
\usepackage{lipsum}
\usepackage{enumerate}
\usepackage{graphicx}
\usepackage{tabularx}
\usepackage{tikz}
\usetikzlibrary{automata, positioning}
\usepackage[left=2cm,right=2cm,top=2cm,bottom=2cm]{geometry}
\RequirePackage[colorlinks,citecolor=blue,urlcolor=blue]{hyperref}

\setattribute{journal}{name}{}

\startlocaldefs
\newcommand{\E}{\mathbb{E}}

\newcommand{\Pt}{\mathbb{P}^{\theta^*}}

\DeclareMathOperator*{\argmax}{arg\,max}
\endlocaldefs

\newtheorem{hyp}{}

\newtheorem{lem}{Lemma}

\newtheorem{thm}{Theorem}

\begin{document}

\begin{frontmatter}
\title{Modeling rainfalls using a seasonal hidden Markov model}
\runtitle{Modeling rainfalls using a seasonal hidden Markov model}
%\thankstext{T1}{Footnote to the title with the ``thankstext'' command.}

\begin{aug}
\author{\fnms{Augustin} \snm{Touron}\thanksref{m1,m2}\ead[label=e1]{augustin.touron@math.u-psud.fr}\ead[label=e2]{augustin.touron@edf.fr}},
%\thankstext{t1}{Some comment}
%\thankstext{t2}{First supporter of the project}
%\thankstext{t3}{Second supporter of the project}
\runauthor{A. Touron}

\affiliation{Universit\' e Paris-Sud\thanksmark{m1} and EDF R\&D\thanksmark{m2}}

\address{Laboratoire de Math\'ematiques d'Orsay,\\ Univ. Paris-Sud, CNRS, Universit\'e Paris-Saclay\\91405 Orsay, France\\
\printead{e1}\\
\phantom{E-mail:\ }}
\address{EDF R\&D,\\ 6, Quai Watier\\78400 Chatou Cedex France\\
\printead{e2}\\
\phantom{E-mail:\ }}

\begin{keyword}[class=MSC]
\kwd[Primary ]{62P12}
\kwd{60K35}
\kwd[; secondary ]{60K35}
\end{keyword}

\begin{keyword}
\kwd{sample}
\kwd{\LaTeXe}
\end{keyword}

\end{aug}

\begin{abstract}
In order to reach the supply/demand balance, electricity providers need to predict the demand and production of electricity at different time scales. This implies the need of modeling weather variables such as temperature, wind speed, solar radiation and precipitation. This work is dedicated to a new daily rainfall generator at a single site. It is based on a seasonal hidden Markov model with mixtures of exponential distributions as emission laws. The parameters of the exponential distributions include a periodic component in order to account for the seasonal behaviour of rainfall. We show that under mild assumptions, the maximum likelihood estimator is strongly consistent, which is a new result for such models. The model is able to produce arbitrarily long daily rainfall simulations that reproduce closely different features of observed time series, including  seasonality, rainfall occurrence, daily distributions of rainfall, dry and rainy spells. The model was fitted and validated on data from several weather stations across Germany. We show that it is possible to give a physical interpretation to the estimated states.
\end{abstract}

%\begin{keyword}[class=MSC]
%\kwd[Primary ]{}
%\kwd{}
%\kwd[; secondary ]{}
%\end{keyword}

%\begin{keyword}
%\kwd{}
%\kwd{}
%\end{keyword}

\end{frontmatter}

% AOS,AOAS: If there are supplements please fill:
%\begin{supplement}[id=suppA]
%  \sname{Supplement A}
%  \stitle{Title}
%  \slink[doi]{10.1214/00-AOASXXXXSUPP}
%  \sdatatype{.pdf}" 
%  \sdescription{Some text}
%\end{supplement}

\section{Introduction}

\subsection{Context and motivation}
Since electricity still cannot be efficiently stored, electricity providers such as EDF have to adjust closely production and demand. There are various ways of producing electricity: nuclear power, coal, gas, hydro-electricity, wind power, solar power, biomass... One of the consequences of the rise of renewable energies is that the electric power industry gets more and more weather-dependent. Weather variables such as temperature, precipitation, solar radiation or wind speed have a growing impact on both production and demand. For example, if the temperature drops by $1^\circ$C in France in winter, the need in power increases by $1500$ MW, which corresponds roughly to one nuclear reactor, or hundreds of wind turbines. The electric power produced by wind turbines, photovoltaic cells and hydroelectric plants also depends directly on weather conditions. The scale of weather information  needed by power industries  is also evolving from the national scale in a centralized system to the more and more local scale with the ongoing decentralization. Therefore, the weather conditions also need to be taken into account at a local scale. In order to achieve this, electricity providers try to evaluate the impact of weather variables on both production and consumption. One of the tools they use for this purpose is weather generators.

\subsection{Weather generators}

A \emph{stochastic weather generator} \citep{katz1996} is a statistical model used whenever we need to quickly produce synthetic time series of weather variables. For early examples of weather generators, see \cite{richardson81} or \citep{katz1977}. These series can then be used as input for physical models (e.g. electricity consumption models), to study climate change, to investigate on extreme values \citep{yaoming2004}...  A good weather generator produces times series that can be considered \emph{realistic}. By realistic we mean that they mimic the behaviour of the variables they are supposed to simulate, according to various criteria. For example, a temperature generator may need to reproduce daily mean temperatures, the seasonality of the variability of the temperature, its global distribution, the distribution of the extreme values, its temporal dependence structure... and so on. The criteria that we wish to consider largely depend on applications. There are many types of weather generators. They may differ by the time frequency of the output data: hourly, daily... Some of them 	are only concerned by one location, whereas others are supposed to model the spatial correlations between several sites located in a more or less large area. They can focus on a single or multiple variables.\\

An important class of weather generators is composed of \textit{state space models}, where a discrete variable called the state is introduced. We then model the variable of interest conditionally to each state. In the field of climate modeling, the states are sometimes called \emph{weather types}. See \cite{Ailliot2015} for an overview of such models. In some cases, they can correspond to large-scale atmospheric circulation patterns, thus giving a physical interpretation to the model. The states may or may not be observed. In the former case, the states are defined \emph{a priori} using classification methods, from the local variables that we wish to model, or from large scale atmospheric quantities such as geopotential fields. If the states are not observed, they are called hidden states. Considering hidden states offers a greater flexibility because the determination of the states is data driven instead of being based on arbitrarily chosen exogeneous variables. Moreover, it is possible to interpret the states \emph{a posteriori} by computing the most likely state sequence. For these reasons, state space models are widely used in climate modeling. Apart from mixture models, the simplest example of hidden state space models is hidden Markov models (HMM): the state process is a Markov chain and the observations are independent conditionally to the states.

\subsection{Modeling precipitation}

Unlike other weather variables, the distribution of daily precipitation amounts naturally appears as a mixture of a mass at $0$ corresponding to dry days, and a continuous distribution with support in $\mathbb{R}_+$ corresponding to the intensity of precipitations on rainy days. Thus we can consider two \emph{states}: a dry state and a wet state \citep{wilks98}. Therefore it is very common to model precipitation using state space models. We can also refine the model by adding sub-states to the wet state, for example a \emph{light rain} state versus a \emph{heavy rain} state. When modeling precipitation, one can focus only on the \emph{occurrence process} \citep{Zucchini91} or on both occurrence and amounts of precipitation \citep{ailliot2009}. In both cases, hidden Markov models have been used extensively. In \cite{bellone2000}, the authors simulate precipitation amounts by using a non-homogeneous hidden Markov model in which the underlying transition probabilities depend on large-scale atmospheric variables. In \citep{lambert2003}, a two-states (wet and dry) non-parametric hidden Markov model is used for precipitation amounts.\\

\paragraph{Seasonality}
Using a simple HMM to model precipitation is not possible because the process of precipitation amounts is not stationary, it exhibits a seasonal behaviour with an annual cycle, like most, if not all, weather variables. Figure \ref{saison_mois} shows the mean of monthly precipitation amounts for the station of Bremen in Germany. 

\begin{figure}
\centering
\includegraphics[scale=0.5]{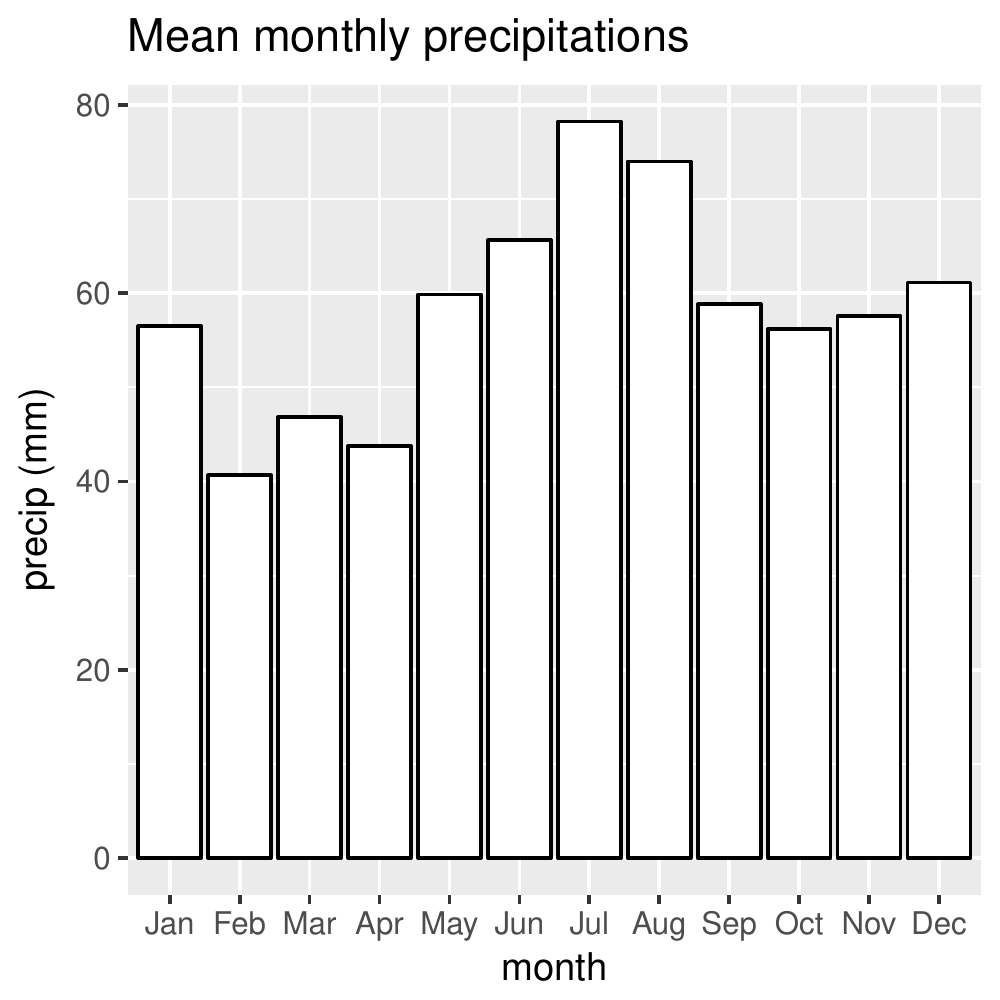}
\caption{Mean precipitation amounts by month}
\label{saison_mois}
\end{figure}

Thus it is necessary to account for this seasonality in our model. There are several ways to do so.
\begin{itemize}
\item In the literature, the most common way to handle seasonality is to split each year into several time periods (e.g. twelve months or four seasons) and to assume that the process to be modelled is stationary within each period. Thus we fit a different model for each time period or we focus on one specific period (e.g. the month of January, or winter). For example in \cite{lennartsson2008}, the authors consider blocks of lengths one, two or three months. See also \cite{Ailliot2015} and references therein. This method is simple to implement as it does not require any further modeling effort, and it may be effective in some cases. However, it has several drawbacks.
\begin{itemize}
\item The stationarity assumption over each period may not be satisfied.
\item We have to fit independently several sub-models, which requires a lot of data.
\item The time series used to fit each of the sub-models is obtained by concatenation of data that do not belong to the same year. For exemple, if the time periods are months, the 31st of January of year $n$ will be followed by the first of January of year $n+1$. This is a problem if we use a Markovian model, which exhibits time dependence.
\item We should be able to simulate a full year using only one model. 
\end{itemize}
\item In time series modeling, another widely spread approach is preprocessing the data in order to obtain a stationary residual. For example, consider $Y_t=s(t)\varepsilon_t$ where $s$ is a deterministic periodic function and $(\varepsilon_t)_t$ a stationary process with unit variance. In this case, we first find some estimator $\hat{s}$ of $s$. Then we model the residual $\frac{Y_t}{\hat{s}(t)}$. This has been adressed in \cite{lambert2003}. Thus the times series appears as a combination of a deterministic part corresponding to seasonality, and a stochastic part. However, it may be difficult to find the right decomposition for $Y_t$.
\item We can avoid splitting the data or preprocessing them by incorporating seasonal parameters in the model. Although it increases the number of parameters and thus the complexity of the estimation problem, this solution offers a greater flexibility by allowing different seasonalities for each state. Therefore, this is the choice we make for the rest of this paper.
\end{itemize}

\subsection{Our contribution}

\paragraph{}In this paper we introduce a state space model for daily precipitation amounts at a single site. Unlike most existing stochastic precipitation generators, our model includes seasonal components, which makes it easy to simulate precipitation time series of arbitrary length. We provide the first theoretical guarantees for the consistency of the maximum likelihood estimator of a seasonal hidden Markov model. The proof relies on the introduction of a suitable stationary hidden Markov model on which we can apply existing consistency results.

\paragraph{}Section 2 deals with our precipitation model from a theoretical point of view. We first give its mathematical formulation, then we state the consistency result (see Theorem \ref{mainthm}) and we prove it. We show that this theorem can be applied to our model. In Section 3 we fit the model to precipitation data. We show that we can easily interpret the estimated states and that the synthetic precipitation time series obtained by simulating according to the model are consistent with observations, which validates the model. The last section is a discussion about how our model is a starting point for further works.

\section{Model}\label{model}

In this section we introduce our model and we study the consistency of the maximum likelihood estimator (MLE).

\subsection{Model description}

\subsubsection{Hidden Markov models}

First, we recall a general definition of finite state-space hidden Markov models. Let $K$ be a positive integer, $\mathsf{X}=\{1,\dots,K\}$ and $(\mathsf{Y},\mathcal{Y})$ a measurable space. A \emph{hidden Markov model} (HMM) with state space $\mathsf{X}$ is a $\mathsf{X}\times\mathsf{Y}$-valued stochastic process $(X_t,Y_t)_{t\geq 1}$ where $(X_t)_{t\geq 1}$ is a Markov chain and $(Y_t)_{t\geq 1}$ are $\mathsf{Y}$-valued random variables that are independant conditonnally on $(X_t)_{t\geq 1}$ and such that for all $j\geq 1$, the conditionnal distribution of $Y_j$ given $(X_t)_{t\geq 1}$ only depends on $X_j$. The law of the Markov chain $(X_t)_{t\geq 1}$ is determined by its initial distribution $\pi$ and its transition matrix $\mathbf{Q}$. For all $k\in\mathsf{X}$, the distribution of $Y_1$ given $X_1=k$ is called the \emph{emission distribution} in state $k$. The Markov chain $(X_t)$ is called the \emph{hidden} Markov chain because it is not accessible to observation. The process $(Y_t)_{t\geq 1}$ only is observed. The two main problems regarding such models are the inference of the parameters based only on the observations, and the determination of a likely state sequence given the parameters. The second problem can be solved using the forward-backward algorithm (see \cite{rabiner86} for an introduction and \cite{cappe2009} for a more modern and general formulation), and the first problem can be solved using the EM algorithm \citep{dempster77}, in which the E step is actually the forward-backward algorithm. In the context of HMM, the EM algorithm is sometimes called the Baum-Welch algorithm \citep{baum1970}.

\subsubsection{Formulation of the model}\label{formulation}

Let $K$, $M$, $d$ be positive integers such that $M\geq 2$. We will now describe the model denoted by $\mathcal{M}(K,M,d)$. Let $\large(X_t^{(1)}\large)_{t\geq 1}$ be a first order homogeneous Markov chain with state space $\mathsf{X} = \{1,\dots,K\}$. Let $\pi$ be its initial distribution and $\mathbf{Q}$ its transition matrix. Let $\mathsf{Y}=[0,+\infty)$ be the observation space, equiped with its Borel $\sigma$-algebra. The observation process $(Y_t)_{t\geq 1}$ is such that the random variables $Y_t$ are independant conditionally on $\large(X_t^{(1)}\large)_{t\geq 1}$ and for all $j\geq 1$, the conditional distribution of $Y_j$ given $\large(X_t^{(1)}\large)_{t\geq 1}$ only depends on $X_j^{(1)}$. The emission distribution $\nu_k$ can be written as
$$\nu_k=p_k\delta_0+(1-p_k)f_k\cdot\boldsymbol{\lambda},$$
where $\delta_0$ refers to the Dirac measure at $0$ and $\boldsymbol{\lambda}$ is the Lebesgue measure on $[0,+\infty)$. The parameter $p_k$ 
is the probability to observe a dry day when in state $k$, and $f_k$ is a probability density function accounting for the intensity of rainfalls. For any positive measurable function $g$ and any positive measure $\mu$, we denote by $g\cdot\mu$ the measure whose density with respect to $\mu$ is $g$. We then have to choose a model for the emission densities $f_k$. Looking at the data, it appears that the distribution of precipitation amounts is very asymetric, with most of the values near $0$ and few large values.\\

Common choices for precipitation modeling are exponential distributions, gamma distributions, or mixtures of those distributions (see e.g. \cite{wilks98} for mixture of exponential distributions and \cite{kenabatho2012} for mixtures of gamma distributions).  When focusing on extreme values, one can also use heavy tail distributions, such as the (generalized) Pareto distribution \citep{lennartsson2008}. In this paper, we will give our consistency results for mixtures of exponential distributions, but they remain valid for various emission distributions, such as mixtures of gammas. For any positive $\lambda$, we denote by $\mathcal{E}(\lambda)$ the exponential distribution with parameter $\lambda$.

To account for seasonality, we shall write the emission distribution in state~$k$ and at time $t$ as

$$\nu_{k,t}:=p_k\delta_0+(1-p_k)\frac{1}{s_k(t)}f_k\left(\frac{\cdot}{s_k(t)}\right)\cdot\boldsymbol{\lambda},$$

\noindent where $s_k$ is a deterministic periodic function acting as a scale parameter. We will assume it is a trigonometric polynomial with degree $d$ and period $T$ (in practice, $T=365$ days):

$$s_k(t) = 1+ \sum_{l=1}^d\left[a_{kl}\cos\left(\frac{2\pi}{T}lt\right)+b_{kl}\sin\left(\frac{2\pi}{T}lt\right)\right].$$

Choosing a fixed value for the constant term is necessary to ensure identifiability. For $t\geq 1$ and $1\leq k\leq K$, we define: $$Z(t):=\begin{pmatrix}
 \cos\left(\frac{2\pi}{T}t\right) & \sin\left(\frac{2\pi}{T}t\right) & \dots & \cos\left(\frac{2\pi}{T}dt\right) & \sin\left(\frac{2\pi}{T}dt\right)
\end{pmatrix}$$ and $$\beta_k:=\begin{pmatrix}
a_{k1} & b_{k1} & \dots & a_{kd} & b_{kd}
\end{pmatrix}^\intercal,$$ so that $s_k(t) = 1 + Z(t)\beta_k$.

Setting $p_{k1}:=p_k$ and $f_k$ being a mixture of exponential densities, the emission distributions are given by:
\begin{equation}\label{emdist}
Y_t\mid\{X_t^{(1)}=k\}\sim p_{k1}\delta_0 + \sum_{m=2}^Mp_{km}\mathcal{E}\left(\frac{\lambda_{km}}{1+Z(t)\beta_k}\right)
\end{equation}

with, for all $k\in\{1,\dots, K\}$, $\sum_{m=1}^Mp_{km}=1$. Thus the parameters of our model are:

\begin{itemize}
\item The initial distribution $\pi$ of the Markov chain $\left(X_t^{(1)}\right)_{t\geq 1}$, considered as a vector in $[0,1]^K$.
\item Its transition matrix $\mathbf{Q}\in\mathbb{R}^{K\times K}$.
\item The weights of the mixture $\mathbf{p}=(p_{km})\in\mathbb{R}^{K\times M}$.
\item The parameters of the exponential distributions $\boldsymbol{\Lambda}=(\lambda_{km})\in\mathbb{R}^{K\times (M-1)}$.
\item The coefficients of the trigonometric polynomials $\boldsymbol{\beta}=(\beta_k)_k\in\mathbb{R}^{K\times (2d)}$.
\end{itemize}

Let $\theta_Y:=(\mathbf{p},\boldsymbol{\Lambda},\boldsymbol{\beta})$ be the vector of parameters of the emission distributions. Equation \eqref{emdist} implies that the  distribution of $Y_t$ given $X_t^{(1)}=k$ is absolutely continuous with respect to the measure $\mu:=\delta_0+\boldsymbol{\lambda}$, and that its density is given by:
$$f_{k,t}^{\theta_Y}(y):=p_{k1}\mathbf{1}_{y=0}+\sum_{m=2}^Mp_{km}\frac{\lambda_{km}}{1+Z(t)\beta_k}\exp\left(- \frac{\lambda_{km}}{1+Z(t)\beta_k}y\right)\mathbf{1}_{y>0}.$$

Stricto sensu, this is not a hidden Markov model because the emission distributions depend not only on the state but also on time. Thus this is a seasonal HMM. Yet one can easily retrieve a HMM by going into higher dimensions. For $j\geq 0$, let us define:
\begin{equation}\label{UjWj}
U_j:=\left(X_{Tj+1}^{(1)},\dots,X_{Tj+T}^{(1)}\right),\quad W_j:=\left(Y_{Tj+1},\dots,Y_{Tj+T}\right).
\end{equation}

As $Z(t)\beta_k$ is $T$-periodic for each $k$, one can show that $(U_j,W_j)_{j\geq 0}$ is a HMM with state space $\{1,\dots,K\}^T$ and observation space $\mathbb{R}_+^T$.

\subsection{Consistency of the maximum likelihood estimator}\label{cons}

Let $\theta:=(\mathbf{Q},\theta^Y)$ be the full vector of parameters. We are working in a parametric framework: $\theta\in\Theta$ where $\Theta$ is a subset of some finite dimensional space. We will assume that $\Theta$ is compact. Moreover, we assume that the model is well specified, i.e. there exists a \emph{true} parameter $\theta^*=\left(\mathbf{Q}^*,\mathbf{p}^*,\boldsymbol{\Lambda}^*,\boldsymbol{\beta}^*\right)$ such that the data is generated by the seasonal HMM with parameter $\theta^*$, and that $\theta^*$ lies in the interior of $\Theta$. In order to fit the model, we shall use the maximum likelihood approach. Recall that the maximum likelihood estimator is defined by
$$\hat{\theta}_{n,\pi}\in\argmax_{\theta\in\Theta}L_{n,\pi}[\theta;(Y_1,\dots,Y_n)],$$
where $(Y_1,\dots,Y_n)$ is the vector of observations and $L_{n,\pi}$ is the likelihood function when the initial distribution of the hidden Markov chain is $\pi$. We shall denote by $\mathbb{P}^{\theta,\pi}$ the law of the process $\left(Y_t\right)_{t\geq 1}$ when the parameter is $\theta$ and the initial distribution is $\pi$. We simply write $\mathbb{P}^\theta$ if $\pi$ is the stationary distribution of $\mathbf{Q}$. Let $\mathbb{E}^{\theta,\pi}(\cdot)$ and $\mathbb{E}^{\theta}(\cdot)$ the corresponding expected values.\\

In this paragraph, we show that, provided that some mild assumptions are satisfied, our model is identifiable and the maximum likelihood estimator is strongly consistent. That is, we shall prove the following theorem:

\begin{thm}\label{mainthm}
Under Assumptions \ref{hyp-irred} to \ref{hyp-beta}, for any initial distribution $\pi$,
$$\Pt\mathrm{-a.s.},\quad\lim\limits_{n\to\infty}\hat{\theta}_{n,\pi}=\theta^*.$$
\end{thm}

Assumptions \ref{hyp-irred} to \ref{hyp-ppos} are defined in Section \ref{ident} and Assumptions \ref{hyp-delta} to \ref{hyp-beta} are defined in Section \ref{cons}.

\subsubsection{Identifiability}\label{ident}

The first step of the proof consists in showing that $\mathbb{P}^{\theta} = \mathbb{P}^{\theta^*}\implies\theta = \theta^*$. In other words, we can retrieve the parameters if we know the law of the process. The proof of identifiability is based on a spectral method. If the law of the process is known, we know the law of $(Y_t,Y_{t+1},Y_{t+2})$ for all $t\in\{1,\dots,T\}$. We will show that this implies that we can retrieve the transition matrix $\mathbf{Q}^*$ and the emission distributions for each $t\in\{1,\dots,T\}$, then for all $t\geq 1$ by periodicity. Finally, we show that the knowledge of the emission distributions implies the knowledge of their parameters. We will need the following assumptions:

\begin{hyp}\label{hyp-irred}
$\mathbf{Q}^*$ is irreducible and its unique stationary distribution $\pi^*$ is the distribution of $X_1^{(1)}$.
\end{hyp}

\begin{hyp}\label{hyp-inv}
$\mathbf{Q}^*$ is invertible.
\end{hyp}

\begin{hyp}\label{hyp-ind}
For all $k\in\{1,\dots,K\}$, $\lambda_{k2}^*<\dots <\lambda_{kM}^*$, and for all $t\in\{1,\dots,T\}$, the $K$ emission distributions are linearly independant.
\end{hyp}

\begin{hyp}\label{hyp-ppos}
For all $k\in\{1,\dots, K\}$ and $m\in\{2,\dots, M\}$, $p_{km}^*>0$.
\end{hyp}

\paragraph{Comments} It is important to notice that these assumptions only involve the \emph{true} parameter $\theta^*$. They may not be satisfied for every $\theta\in\Theta$. The Assumptions \ref{hyp-irred}, \ref{hyp-inv} and \ref{hyp-ppos} are obviously generically satisfied. The following lemma gives a sufficient condition for the genericity of \ref{hyp-ind}. 

\begin{lem}\label{lem-ind}
Assumption \ref{hyp-ind} is generically satisfied if for all $t\in\{1,\dots,T\}$, the cardinality of the set $\left\lbrace\frac{\lambda^*_{km}}{s_k(t)}:1\leq k\leq K,\, 2\leq m\leq M\right\rbrace$ is at least $K$.
\end{lem}
\begin{proof}
For $t\in\{1,\dots,T\}$, let $\nu_k(t):=\sum_{m=2}^Mp_{km}^*\mathcal{E}\left(\frac{\lambda^*_{km}}{s_k(t)}\right)$. As the measures $\delta_0$ and $\boldsymbol{\lambda}$ are mutually singular, it suffices to show that the measures $\left(\nu_k(t)\right)_{1\leq k\leq K}$ are linearly independent. We define the set:
$$E_t:=\left\lbrace\frac{\lambda^*_{km}}{s_k(t)}:1\leq k\leq K,\, 2\leq m\leq M\right\rbrace$$
and $p(t)$ its cardinality. Thus $K\leq p(t)\leq K(M-1)$ and we can write $E_t=\{\tilde{\lambda}^*_1,\dots,\tilde{\lambda}^*_{p(t)}\}$. For all $k\in\{1,\dots K\}$, $\nu_k(t)$ is a linear combination of exponential distributions whose parameters belong to $E_t$. Hence,
$$\nu_k(t) = \sum_{j=1}^{p(t)}B_{kj}(t)\mathcal{E}(\tilde{\lambda}_j^*),$$
where the $B_{kj}(t)$ are among the $p_{km}^*$. The $\tilde{\lambda}_j^*$ being pairwise distinct, the family $\left(\mathcal{E}(\tilde{\lambda}_j^*)\right)_{1\leq j\leq p(t)}$ is linearly independent. Hence, the family $\left(\nu_k(t)\right)_{1\leq k\leq K}$ is linearly independent if and only if the rank of the matrix $B(t)$ is $K$ (this requires that $p(t)\geq K$). This holds true except if the family $(p^*_{km})$ belongs to a set of roots of a polynomial. Moreover, the entries of $B(t)$ are among the $p^*_{km}$. This implies that the range of the map $t\mapsto B(t)$ is finite. Thus we obtain linear independence of the $\left(\nu_k(t)\right)_{1\leq k\leq K}$ for all $t$, except if the $p^*_{km}$ belong to a finite union of roots of polynomials.
\end{proof}

\paragraph{}The following proof is based on the spectral algorithm as presented in \cite{hsu2012} (see also \cite{DCGL2016} and \cite{DCGLC2017}). Let $\left(N_r\right)_{r\geq 1}$ be an increasing sequence of positive integers and $\left(\mathfrak{P}_{N_r}\right)_{r\geq 1}$ an increasing sequence of subspaces of $L^2(\mathbb{R}_+,\mu)$ whose union is dense in $L^2(\mathbb{R}_+,\mu)$. Let $(\phi_1,\dots,\phi_{N_r})$ an orthonormal basis of $\mathfrak{P}_{N_r}$ such that for any $f\in L^2(\mathbb{R}_+,\mu)$, 
$$\lim\limits_{r\to\infty}\left\|\sum_{n=1}^{N_r}\langle\phi_n,f\rangle\phi_n-f\right\|_2=0,$$
where $\langle\cdot,\cdot\rangle$ denotes the inner product in $L^2(\mathbb{R}_+,\mu)$ and $\|\cdot\|_2$ the euclidean norm. Since $f^{\theta^*_Y}_{k,t}\in L^2(\mathbb{R}_+,\mu)$ for any $k\in\{1,\dots,K\}$ and $t\geq 1$, one can retrieve $f^{\theta^*_Y}_{k,t}$ from its projections on the spaces $\mathfrak{P}_{N_r}$. Let $r\geq 1$ and $N:=N_r$. Before going through the spectral algorithm, let us first introduce some notations. For $t\geq 2$, we shall consider the following vectors, matrices and tensors:
\begin{itemize}
\item[$\bullet$] $g_t$ the probability density function of $(Y_{t-1},Y_t,Y_{t+1})$,
\item[$\bullet$] $\mathbf{L}_{t}\in\mathbb{R}^N$ the vector defined by $\mathbf{L}_{t}(a):=\E[\phi_a(Y_{t})]$  ,
\item[$\bullet$] $\mathbf{M}_t\in\mathbb{R}^{N\times N\times N}$ the tensor such that $\mathbf{M}_t(a,b,c) := \E[\phi_a(Y_{t-1})\phi_b(Y_t)\phi_c(Y_{t+1})]=\langle g_t,\phi_a\otimes\phi_b\otimes\phi_c\rangle $,
\item[$\bullet$] $\mathbf{N}_t\in\mathbb{R}^{N\times N}$ the matrix defined by $\mathbf{N}_t(a,b) := \E[\phi_a(Y_{t})\phi_b(Y_{t+1})]$,
\item[$\bullet$] $\mathbf{P}_t\in\mathbb{R}^{N\times N}$ the matrix defined by $\mathbf{P}_t(a,c) := \E[\phi_a(Y_{t-1})\phi_c(Y_{t+1})]$,
\item[$\bullet$] $\mathbf{O}_t\in\mathbb{R}^{N\times K}$ the matrix defined by $\mathbf{O}_t(a,k) := \mathbb{E}[\phi_a(Y_t)\mid X_t = k]=\langle f_{k,t}^{\theta_Y^*},\phi_a\rangle$.
\end{itemize}
Apart from $\mathbf{O}_t$, these quantities can be computed from the law of $(Y_{t-1},Y_t,Y_{t+1})$. Using the previous definitions, one easily proves the following equalities:

\begin{lem}\label{lem-spec}
\begin{align}
\label{eq1}
\mathbf{L}_t &= \mathbf{O}_t\pi^*\\
\label{eq2}
\forall b\in\{1,\dots,N\},\quad \mathbf{M}_t(\cdot,b,\cdot) &= \mathbf{O}_{t-1}\mathrm{diag}(\pi^*)\mathbf{Q}^*\mathrm{diag}[\mathbf{O}_t(b,\cdot)]\mathbf{Q}^*\mathbf{O}_{t+1}^T\\
\label{eq3}
\mathbf{N}_t &= \mathbf{O}_t\mathrm{diag}(\pi^*)\mathbf{Q}^*\mathbf{O}_{t+1}^T\\
\label{eq4}
\mathbf{P}_t &= \mathbf{O}_{t-1}\mathrm{diag}(\pi^*)(\mathbf{Q}^*)^2\mathbf{O}_{t+1}^T 
\end{align}
where $\mathrm{diag(v)}$ is the diagonal matrix whose diagonal entries are the entries of the vector $v$.
\end{lem}

Notice that \ref{hyp-irred} implies that all the entries of $\pi^*$ are positive. Thanks to Assumption \ref{hyp-ind}, for a large enough $N>K$ , the matrices $\mathbf{O}_t$ have rank $K$. In addition, as $\mathbf{Q}^*$ is full rank (Assumption  \ref{hyp-inv}), equations \eqref{eq3} et \eqref{eq4} show that the matrices $\mathbf{P}_t$ and $\mathbf{N}_t$ also have rank $K$.

\paragraph{}
Let $\mathbf{P}_t = \mathbf{U}\mathbf{\Sigma}\mathbf{V}^T$ be a singular value decomposition of $\mathbf{P}_t$: $\mathbf{U}$  and $\mathbf{V}$ are matrices of size $N\times K$ whose columns are orthonormal families being the left (resp. right) singular vectors of $\mathbf{P}_t$ associated with its $K$ non-zero singular values, and $\mathbf{\Sigma}=\mathbf{U}^T\mathbf{P}_t\mathbf{V}$ is an invertible diagonal matrix of size $K$ containing these singular values. Let us define, for $1\leq b\leq N$: 
$$\mathbf{B}(b):=(\mathbf{U}^T\mathbf{P}_t\mathbf{V})^{-1}\mathbf{U}^T\mathbf{M}_t(\cdot,b,\cdot)\mathbf{V}.$$
Using Lemma \ref{lem-spec}, we obtain:
\begin{align*}
\mathbf{B}(b) &= \left(\mathbf{U}^T\mathbf{O}_{t-1}\mathrm{diag}(\pi^*)(\mathbf{Q^*})^2\mathbf{O}_{t+1}^T\mathbf{V}\right)^{-1}\mathbf{U}^T\mathbf{O}_{t-1}\mathrm{diag}(\pi^*)\mathbf{Q}^*\mathrm{diag}[\mathbf{O}_t(b,\cdot)]\mathbf{Q}^*\mathbf{O}_{t+1}^T\mathbf{V}\\
&= \left(\mathbf{O}_{t+1}^T\mathbf{V}\right)^{-1}(\mathbf{Q}^*)^{-1}\mathrm{diag}[\mathbf{O}_t(b,\cdot)]\mathbf{Q}^*\mathbf{O}_{t+1}^T\mathbf{V}\\
&= \left(\mathbf{Q}^*\mathbf{O}_{t+1}^T\mathbf{V}\right)^{-1}\mathrm{diag}[\mathbf{O}_t(b,\cdot)]\left(\mathbf{Q}^*\mathbf{O}_{t+1}^T\mathbf{V}\right).
\end{align*}

Hence the matrix $\mathbf{R}:=\left(\mathbf{Q}^*\mathbf{O}_{t+1}^T\mathbf{V}\right)^{-1}$ diagonalizes all the matrices $\left(\mathbf{B}(b)\right)_{1\leq b\leq N}$. Let $\tilde{\mathbf{U}}\tilde{\mathbf{\Sigma}}\tilde{\mathbf{V}}$ a singular value decomposition of $\mathbf{N}_t$ and for $1\leq k\leq K$, we define:

$$\mathbf{C}(k) := \sum_{b=1}^N\tilde{\mathbf{U}}(b,k)\mathbf{B}(b).$$

Thus we have:

$$\mathbf{R}^{-1}\mathbf{C}(k)\mathbf{R} = \sum_{b=1}^N\tilde{\mathbf{U}}(b,k)\mathrm{diag}[\mathbf{O}_t(b,\cdot)] = \mathrm{diag}[\tilde{\mathbf{U}}^T\mathbf{O}_t(k,\cdot)].$$

Let $\boldsymbol{\Lambda}$ the $K\times K$ matrix defined by $\boldsymbol{\Lambda}(k,k') := [\mathbf{R}^{-1}\mathbf{C}(k)\mathbf{R}](k',k')$. Hence $$\boldsymbol{\Lambda}=\tilde{\mathbf{U}}^T\mathbf{O}_t.$$ Therefore, from the equality $\tilde{\mathbf{U}}\tilde{\mathbf{U}}^T\mathbf{O}_t = \mathbf{O}_t$, we get:

$$\mathbf{O}_t = \tilde{\mathbf{U}}\boldsymbol{\Lambda}.$$

It follows from the above that for all $t\in\{1,\dots T\}$, the matrix $\mathbf{O}_t$ is computable from $\mathbf{M}_t$, $\mathbf{N}_t$ and $\mathbf{P}_t$. As $\bigcup_{r\geq 1}\mathfrak{P}_{N_r}$ is dense in $L^2(\mathbb{R}_+,\mu)$, this implies the knowledge of the emission distibutions for every $t\geq 1$, $Y_1$ and $Y_{T+1}$ having the same distribution thanks to periodicity and Assumption \ref{hyp-irred}. Then we notice that:
$$(\tilde{\mathbf{U}}^T\mathbf{O}_t)^{-1}\tilde{\mathbf{U}}^T\mathbf{L}_t = (\tilde{\mathbf{U}}^T\mathbf{O}_t)^{-1}\tilde{\mathbf{U}}^T\mathbf{O}_t\pi^* = \pi^*.$$

We finally obtain the transition matrix:

\begin{align*}
\left(\tilde{\mathbf{U}}^T\mathbf{O}_t\mathrm{diag}(\pi^*)\right)^{-1}\tilde{\mathbf{U}}^T\mathbf{N}_t\mathbf{V}\left(\mathbf{O}_{t+1}^T\mathbf{V}\right)^{-1} 
&= \mathbf{Q}^*.
\end{align*}

We now show that we can identify $\theta_Y^*$ from the emission densities $f_{k,t}^{\theta^*_Y}$.

\paragraph{}Let us first prove that we can obtain the seasonalities $s_k(\cdot)=1+Z(\cdot)\beta_k^*$. Recall that for all $k\in\{1,\dots K\}$, there exists a random variable $\tilde{Y}_k$ such that for all $t\geq 1$, $Y_t\mid \{X_t^{(1)}=k\} \sim s_k(t)\tilde{Y}_k$. Denoting by $\mathbb{V}(X)$ the variance operator, we have:
$$\tilde{s}(t):=\frac{s_k(t)}{s_k(1)}=\sqrt{\frac{\mathbb{V}(Y_t\mid X_t^{(1)}=k)}{\mathbb{V}(Y_1\mid X_1^{(1)}=k)}}.$$
$\tilde{s}(t)$ can be computed for every time $t$ from the emission densities. As it is a trigonometric polynomial, we can write $\tilde{s}(t) = c + T(t)$ where $c$ is its constant coefficient, that we can identify. Observing that the constant coefficient of~$s_k$ is $1$, we get $s_k(t) = \frac{\tilde{s}(t)}{c}$, from which we deduce $\beta_k^*$.

\paragraph{}Now we are left with the identification of the parameters of the distribution $p_{k1}^*\delta_0+\sum_{m=2}^Mp_{km}^*\mathcal{E}(\lambda_{km}^*)$. Let $f_k :y\mapsto p_{k1}^*\mathbf{1}_{y=0} + \sum_{m=2}^Mp_{km}^*\lambda_{km}^*\exp(-\lambda_{km}^*y)\mathbf{1}_{y>0}$ the corresponding density. 

First, $p_{k1}^*=f_k(0)$. Besides, by Assumption \ref{hyp-ind}, we have $\lambda_{k2}^*<\cdots <\lambda_{kM}^*$. It follows that:
$$f_k(y)\sim_{+\infty}p_{k2}^*\lambda_{k2}^*\exp(-\lambda_{k2}^*y),$$
thus $$\log f_k(y)\sim_{+\infty}\log p_{k2}^* + \log \lambda_{k2}^* - \lambda_{k2}^*y.$$
Hence $\lambda_{k2}^*=-\lim\limits_{y\to\infty}\frac{\log f_k(y)}{y}$ and $p_{k2}^*=\exp\left(\lim\limits_{y\to\infty}(\log f_k(y) + \lambda_{k2}^*y-\log\lambda_{k2}^*)\right)$. Finally, step by step, we identify in the same way the remaining parameters, which concludes the proof of identifiability.

\paragraph{Remarks}
\begin{itemize}
\item The model is only identifiable up to permutation of the states.
\item The spectral algorithm provides a way to estimate the transition matrix and the emission distributions in a nonparametric framework \citep{hsu2012}.
\item Since our proof of identifiability is constructive, one could use it to build other estimators of the parameters, for example using methods of moments.
\end{itemize}

\subsubsection{Strong consistency}\label{cons}

Recall that the parameter of our model is $\theta=\left(\mathbf{Q},\mathbf{p},\boldsymbol{\Lambda},\boldsymbol{\beta}\right)\in\Theta$ (see Section \ref{formulation}). In order to prove the almost sure convergence of the maximum likelihood estimator to $\theta^*$, we will make the following assumptions:

\begin{hyp}\label{hyp-delta}
$\delta:=\inf_{\theta\in\Theta}\min_{i,j}\mathbf{Q}(i,j)>0$.
\end{hyp}

\begin{hyp}\label{hyp-p}
There exists $\bar{p}_{\min}>0$ such that for all $\theta\in\Theta$, 
$$\sum_kp_{k1}\geq\bar{p}_{\min}.$$
\end{hyp}

\begin{hyp}\label{hyp-lambda}
There exists $\lambda_{\max}>\lambda_{\min}>0$ such that for all $k,m$ and for all $\theta\in\Theta$, 
$$\lambda_{km}\in [\lambda_{\min},\lambda_{\max}].$$
\end{hyp}

\begin{hyp}\label{hyp-beta}
There exists $\sigma_{\max}>\sigma_{\min}>0$ such that for all $\theta\in\Theta$, $k\in\{1,\dots,K\}$ and $t\in\{1,\dots T\}$, $$1+Z(t)\beta_k\in [\sigma_{\min},\sigma_{\max}].$$
\end{hyp}

\paragraph{Comments} These assumptions are not needed for identifiability but only for the strong consistency of the MLE. They are uniform (in $\theta$) boundedness conditions on the parameters. In practice, we do not use these restrictions when performing the maximization of the likelihood function but we just ensure that all the entries of $\boldsymbol{\Lambda}$ are positive and that for all $k\in\{1,\dots,K\}$ and $t\in\{1,\dots T\}$, $s_k(t)>0$.

\paragraph{}

Recall that the stochastic process $(U_j,W_j)_{j\geq 0}$ as defined in \eqref{UjWj} is a hidden Markov model with state space $\mathsf{U}:=\{1,\dots,K\}^T$ and observation space $\mathsf{W}:=\mathbb{R}_+^T$. Under the parameter $\theta$ and initial distribution $\pi$, its initial distribution can be written as a function of $\pi$ and $\mathbf{Q}$, its transition matrix $\tilde{\mathbf{Q}}^\theta\in\mathbb{R}^{K^T\times K^T}$ as a function of $\mathbf{Q}$, and for $w=(y_1,\dots,y_T)\in\mathsf{W}$ and $u=(u_1,\dots,u_T)\in\mathsf{U}$, the emission density of $W_0$ given $U_0=u$ is:

$$g^\theta(w\mid u) := \prod_{i=1}^Tf^{\theta_Y}_{u_i,i}(y_i).$$
Hence the law of the HMM $(U_j,W_j)_j$ is entirely determined by the parameter $\theta$ of the process $(X_t,Y_t)_{t\geq 1}$ and its initial distribution. Denoting by $\mathbb{Q}^\theta$ the law of $(W_j)_{j\geq 0}$ when the parameter is $\theta$ and the initial distribution is the stationary distribution, we notice that for any $\theta_1$, $\theta_2\in\Theta$, 
$$\mathbb{Q}^{\theta_1}=\mathbb{Q}^{\theta_2}\implies\mathbb{P}^{\theta_1}=\mathbb{P}^{\theta_2}.$$
Thus, using the conclusions of Paragraph \ref{ident}, we have:
\begin{equation}\label{ident2}
\mathbb{Q}^\theta=\mathbb{Q}^{\theta^*}\implies\theta = \theta^*.
\end{equation}

We can check that for all $\theta\in\Theta$, $J\geq 0$ and for any initial distribution $\pi$, we have
$$\tilde{L}_{J,\pi}[\theta;(W_0,\dots,W_J)] = L_{(J+1)T,\pi}[\theta;(Y_1,\dots,Y_{(J+1)T})],$$
where $L_{n,\pi}$ (resp $\tilde{L}_{J,\pi}$) denotes the likelihood function for the process $(X_t,Y_t)_{t\geq 1}$ (resp. $(U_j,W_j)_{j\geq 0}$) when the initial distribution is $\pi$.
As a consequence, if we denote by $\tilde{\theta}_{J,\pi}$ a maximizer of $\tilde{L}_{J,\pi}$ it suffices to show the strong consistency of $\tilde{\theta}_{J,\pi}$ for any $\pi$ to obtain the desired result. Indeed, if we are able to show the strong consistency of $\tilde{\theta}_{J,\pi}$, we can prove using the same arguments that for any $s\in\{0,\dots,T-1\}$, the estimator $\tilde{\theta}^s_{J,\pi}:=\argmax_{\theta\in\Theta}{L}_{(J+1)T+s,\pi}[\theta;Y_1,\dots,Y_{(T+1)J+s}]$ is strongly consistent. From this we easily derive that $\hat{\theta}_{n,\pi}$ is strongly consistent.

\begin{lem}\label{lem-hyp}
Assuming \ref{hyp-delta}-\ref{hyp-p}-\ref{hyp-lambda}-\ref{hyp-beta}, the following properties hold:
\begin{enumerate}[(i)]
\item  $$\tilde{\delta}:=\inf_{\theta\in\Theta}\inf_{u,v}\tilde{\mathbf{Q}}^\theta(u,v)>0.$$
\item  For all $w\in\mathsf{W}$, $$\inf_{\theta\in\Theta}\sum_u g^\theta (w\mid u)>0,\quad \sup_{\theta\in\Theta}\sum_u g^\theta (w\mid u)<\infty.$$
\item $$b_+:=\sup_{\theta\in\Theta}\sup_{w,u}g^\theta (w\mid u)< \infty,\quad\mathbb{E}^{\theta^*}[ | \log b_-(W_0) |]<\infty,$$ where $b_-(w) :=\inf_{\theta\in\Theta}\sum_u g^\theta(w\mid u)$. 
\item For all $u,v\in\mathsf{U}$, $w\in\mathsf{W}$, the maps $\theta\mapsto \tilde{\mathbf{Q}}^\theta (u,v)$ and $\theta\mapsto g^\theta (w\mid u)$ are continuous.
\end{enumerate}
\end{lem}
\begin{proof}
The first property follows from Assumption \ref{hyp-delta}. Property (iv) clearly holds. The second property and the fact that $b_+< \infty$ follow from Assumptions \ref{hyp-p}-\ref{hyp-lambda}-\ref{hyp-beta}. It remains to prove that $\mathbb{E}^{\theta^*}[ | \log b_-(W_0) |]<\infty$. We have
$$\mathbb{E}^{\theta^*}[ | \log b_-(W_0) |]=\sum_u\tilde{\pi}^{\theta^*}(u)\int g^{\theta^*}(w\mid u)|\log b_-(w)|\mu^{\otimes T}(dw),$$
where $\tilde{\pi}^{\theta^*}$ is the stationary distribution associated with $\tilde{\mathbf{Q}}^{\theta^*}$. Thus it suffices to prove that for all state vector $u\in\mathsf{U}$, 
\begin{equation}\label{eq_b-}
\int g^{\theta^*}(w\mid u)|\log b_-(w)|\mu^{\otimes T}(dw)<\infty.
\end{equation}
Under Assumptions \ref{hyp-p}-\ref{hyp-lambda}-\ref{hyp-beta}, we get $\sup_w b_-(w)<\infty$. Hence \eqref{eq_b-} holds if there exists some function $C$ such that for all $w\in\mathsf{W}$, $b_-(w)\geq C(w)$ and 
\begin{equation}
\int g^{\theta^*}(w\mid u)(-\log C(w))\mu^{\otimes T}(dw)<\infty.
\end{equation}
For $1\leq i\leq T$, let $c_i(y_i):=\inf_{\theta}\sum_kf^{\theta_Y}_{k,i}(y_i)$ and $C(w) := \prod_{i=1}^Tc_i(y_i)$. Expanding this product, we get, for all $\theta\in\Theta$ and $w=(y_1,\dots,y_T)\in\mathsf{W}$,
$$\prod_{i=1}^T\sum_{k=1}^Kf^{\theta_Y}_{k,i}(y_i)=\sum_{(u_1,\dots,u_T)}\prod_{i=1}^Tf^{\theta_Y}_{u_i,i}(y_i)=\sum_ug^\theta(w\mid u).$$
Therefore, for all $w\in\mathsf{W}$, we have $$C(w)=\prod_{i=1}^T\inf_{\theta}\sum_{k=1}^Kf^{\theta_Y}_{k,i}(y_i) \leq \inf_{\theta}\prod_{i=1}^T\sum_{k=1}^Kf^{\theta_Y}_{k,i}(y_i)=b_-(w).$$
On the other hand, Fubini's theorem yields
$$\int g^{\theta^*}(w\mid u)(-\log C(w))\mu^{\otimes T}(dw) = \sum_{i=1}^T\int f^{\theta^*_Y}_{u_i,i}(y_i)(-\log c_i(y_i))\mu(dy_i).$$
Under Assumptions \ref{hyp-p}-\ref{hyp-lambda}-\ref{hyp-beta}, these integrals are finite. Indeed, $c_i(0)\geq\bar{p}_{\min}$ and for $y_i>0$, $c_i(y_i)\geq \bar{p}_{\min}\frac{\lambda_{\min}}{1+\sigma_{\max}}\exp\left[-\frac{\lambda_{\max}}{1+\sigma_{\min}}y_i\right]$, which ends the proof. 
\end{proof}

The conclusions of Lemma \ref{lem-hyp} are in fact sufficient conditions to apply Theorem 13.14 in \cite{douc2014}. Combining this theorem with \eqref{ident2}, we obtain the strong consistency of $\tilde{\theta}_{J,\pi}$ and then Theorem \ref{mainthm} is proved.

\subsection{Computation of the maximum likelihood estimator}\label{fit}

For given number of states $K$, number of populations $M$ and polynomial degree $d$, we wish to estimate the parameters of the model $\mathcal{M}(K,M,d)$ using maximum likelihood inference. We have already shown the strong consistency of the maximum likelihood estimator. This paragraph deals with its practical computation. Before giving the expression of the likelihood function, let us first introduce an equivalent formulation of the model. Let $\left(X_t^{(2)}\right)_{t\geq 1}$ be random variables such that for all $t\geq 1$, the distribution of $X_t^{(2)}\in\{1,\dots, M\}$ conditionally to $\left(X_s^{(1)}\right)_{s\geq 1}$ only depends on $X_t^{(1)}$. This stochastic process models the choice of a population of the mixture, given the hidden state. We have, for all $t\geq 1$, $k\in\{1,\dots K\}$ and $m\in\{1,\dots,M\}$:
$$\mathbb{P}\left(X_t^{(2)}=m\mid X_t^{(1)}=k\right)=p_{km}.$$
Let $X_t=\left(X_t^{(1)},X_t^{(2)}\right)$. Observe that $(X_t)_{t\geq 1}$ is a Markov chain. Using this new state space, the emission distributions write:
$$Y_t\mid \{X_t=(k,m)\}\sim\left\lbrace
\begin{array}{lr}
\delta_0, & m = 1 \\
\mathcal{E}\left(\frac{\lambda_{km}}{s_k(t)}\right), & m\geq 2
\end{array}\right..$$ 
The corresponding emission densities are:
$$f_{k,m,t}^{\theta_Y}(y) := \left\lbrace\begin{array}{lr}
\mathbf{1}_{y = 0}, & m = 1\\
\frac{\lambda_{km}}{s_k(t)}\exp\left(-\frac{\lambda_{km}}{s_k(t)}y\right)\mathbf{1}_{y>0}, & m\geq 2
\end{array}\right..$$
Assume that we have observed a trajectory of the process $(Y_t)_{t\geq 1}$ with length $n$. Let $X:=(X_1,\dots,X_n)$ and $Y:=(Y_1,...,Y_n)$ and recall that $X$ is not observed. The likelihood function is then
$$L_{n}(\theta;Y)=\sum_{\mathbf{x}^{(1)},\mathbf{x}^{(2)}}\pi_{x^{(1)}_1}p_{x^{(1)}_1,x^{(2)}_1}f^{\theta_Y}_{x^{(1)}_1,x^{(2)}_1,t}(Y_1)\prod_{t=2}^n\mathbf{Q}_{x^{(1)}_{t-1}x^{(1)}_t}p_{x^{(1)}_tx^{(2)}_t}f^{\theta_Y}_{x^{(1)}_t,x^{(2)}_t,t}(Y_t),$$
where $\mathbf{x}^{(i)}=\left(x^{(i)}_1,\dots,x^{(i)}_n\right)$ for $i\in\{1,2\}$ and $\pi$ is the stationary distribution corresponding to $\mathbf{Q}$ (its existence is guaranteed by Assumption \ref{hyp-delta}). As $X$ is not observed, we use the Expectation Maximization (EM) algorithm to find a local maximum of the log-likelihood function. The EM algorithm is a classical algorithm to perform maximum likelihood inference with incomplete data. For any initial distribution $\pi$, we define the \emph{complete} log-likelihood by:

\begin{align*}
\log L_{n,\pi}\left[\theta;(X,Y)\right]:=\log \pi_{X_1^{(1)}}+\sum_{t=1}^{n-1}\log \mathbf{Q}_{X_t^{(1)}X_{t+1}^{(1)}}+\sum_{t=1}^n\log p_{X_t^{(1)}X_t^{(2)}}+\sum_{t=1}^n\log f^{\theta_Y}_{X_t^{(1)},X_t^{(2)},t}\left(Y_t\right).
\end{align*}

The algorithm starts from an initial vector of parameters $(\theta^{(0)},\pi^{(0)})$ and alternates between two steps to build a sequence of parameters $\left(\theta^{(q)},\pi^{(q)}\right)_{q\geq 0}.$
\begin{itemize}
\item The \textbf{E} step is the computation of the \emph{intermediate quantity} defined by:
$$Q\left[\left(\theta,\pi\right),\left(\theta^{(q)},\pi^{(q)}\right)\right]:=\mathbb{E}^{\theta^{(q)},\pi^{(q)}}\left[\log L_{n,\pi}\left(\theta;(X,Y)\right)\mid Y\right].$$
\item The \textbf{M} step consists in finding $\left(\theta^{(q+1)},\pi^{(q+1)}\right)$ maximizing the function $\left(\theta,\pi\right)\mapsto Q\left[\left(\theta,\pi\right),\left(\theta^{(q)},\pi^{(q)}\right)\right]$, or at least increasing it.
\end{itemize}
It can be shown that the sequence of likelihoods $L_n\left(\theta^{(q)};Y\right)$ is increasing and that under regularity conditions, it converges to a local maximum of the likelihood function \citep{wu83}. The \textbf{E} step requires the computation of the following quantities:

\begin{itemize}

\item The \emph{smoothing} probabilities: 
$$\pi_{t\mid n}(k):=\mathbb{P}^{\theta^{(q)},\pi^{(q)}}\left(X_t^{(1)}=k\mid Y\right)$$
 for all $k\in\mathsf{X}$ and $1\leq t\leq n$.
\item The bivariate smoothing probabilities:
$$\pi_{t,t+1 \mid n}^{(q)}(k,l) := \mathbb{P}^{\theta^{(q)},\pi^{(q)}}\left(X_t^{(1)} = k,X_{t+1}^{(1)} = l\mid Y\right)$$ for $k,l\in\mathsf{X}$ and $1\leq t\leq n-1$.
\item $\gamma_t^{(q)}(k,m) := \mathbb{P}^{\theta^{(q)},\pi^{(q)}}\left(X_t = (k,m)\mid Y\right)$. 
\end{itemize}
The computation of the smoothing probabilities can be done efficiently using the \emph{forward-backward} algorithm \citep{rabiner86}. In addition, we have:
$$\gamma_t^{(q)}(k,m) = \pi^{(q)}_{t\mid n}(k)\frac{p^{(q)}_{km}f^{(q)}_{km,t}(Y_t)}{\sum_{m'=1}^Mp^{(q)}_{km'}f^{(q)}_{km',t}(Y_t)},$$
where $\left(p^{(q)}_{km}\right)$ and $f^{(q)}_{km,t}$ refer respectively to the weight matrix and the emission densities corresponding to the parameter $\theta^{(q)}$. Then, the intermediate quantity is given by the following formula:
\begin{equation*}
\begin{split}
Q\left[\left(\theta,\pi\right),\left(\theta^{(q)},\pi^{(q)}\right)\right]&=\mathbb{E}^{\theta^{(q)},\pi^{(q)}}\left[\log L_{n,\pi}\left(\theta;(X,Y)\right)\mid Y\right]\\
 &= \sum_{k=1}^K\pi_{1\mid n}^{(q)}(k)\log\pi_k\\
 &+ \sum_{t=1}^{n-1}\sum_{k=1}^K\sum_{l=1}^K\pi_{t,t+1\mid n}^{(q)}(k,l)\log \mathbf{Q}(k,l)\\
 &+ \sum_{t=1}^n\sum_{k=1}^K\sum_{m=1}^M\gamma_t^{(q)}(k,m)\log p_{km}\\
 &+ \sum_{t=1}^n\sum_{k=1}^K\sum_{m=1}^M\gamma_t^{(q)}(k,m)\log f^{\theta_Y}_{km,t}(Y_t).
\end{split}
\end{equation*}
Each of these four terms can be maximized separately, and we can perform the maximization state by state. The constrained optimization for the first three terms yields:
\begin{align*}
\pi_k^{(q+1)} &= \pi_{1\mid n}^{(q)},\quad 1\leq k\leq K\\
\mathbf{Q}^{(q+1)}(k,l) &= \frac{\sum_{t=1}^{n-1}\pi^{(q)}_{t,t+1\mid n}(k,l)}{\sum_{t=1}^{n-1}\pi^{(q)}_{t\mid n}(k)},\quad 1\leq k,l\leq K\\
p_{km}^{(q+1)} &=\frac{\sum_{t=1}^n\gamma^{(q)}_t(k,m)}{\sum_{t=1}^n\pi^{(q)}_{t\mid n}(k)},\quad 1\leq k\leq K, 1\leq m\leq M.
\end{align*}
We then have to maximize with respect to $\theta_Y$, for each $k$:
\begin{align*}
\sum_{t=1}^n\sum_{m=1}^M\gamma^{(q)}_t(k,m)\log f^{\theta_Y}_{km}(Y_t) = \sum_{t=1}^n\sum_{m=2}^M\gamma_t(k,m)\log f^{\theta_Y}_{km}(Y_t)\\
 = \sum_{t=1}^n\sum_{m=2}^M\gamma^{(q)}_t(k,m)\left(\log\lambda_{km} - \log (1+Z(t)\beta_k) -\frac{\lambda_{km}}{1+Z(t)\beta_k}Y_t\right),\\
\end{align*}
under the constraints $\lambda_{km}>0$ and $1+Z(t)\beta_k>0$. This requires to use one of the many numerical optimization algorithms, keeping in mind that the objective function is not convex. We alternate the two steps of the EM algorithm until we reach a stopping criterion. For example, we can stop the algorithm when the relative difference $\frac{L_{n,\pi^{(q)}}(\theta^{(q+1)};Y)-L_{n,\pi^{(q)}}(\theta^{(q)};Y)}{L_{n,\pi^{(q)}}(\theta^{(q)};Y)}$ drops below some threshold $\varepsilon$. The last computed term of the sequence $\left(\theta^{(q)}\right)_{q\geq 0}$ is then an approximation of the maximum likelihood estimator. However, if the EM algorithm does converge, it only guarantees that the limit is a \emph{local} maximum of the likelihood function, which may not be global. Therefore it is a common practice to run the algorithm a large number of times, starting from different (e.g. randomly chosen) initial points and select the parameter with the largest likelihood. In \cite{biernacki2003}, the authors compare several procedures to initialize the EM algorithm, using variants such as SEM \citep{broniatowski1983}. Introducing randomness in the EM algorithm provides a way to escape from local maxima.

\paragraph{Remark} In the computation of the conditional expectancy, we should only sum over the pairs $(t,m)$ such that $\gamma_t(k,m)>0$. Notice that provided that the $p_{km}$ are positive, we have $\gamma_t(k,m)=0$ if and only if $f^{\theta_Y}_{km,t}(Y_t)=0$. Therefore, we sum over all the indices, using the convention $0\times(-\infty)=0$ in the case where $\gamma_t(k,m)=f^{\theta_Y}_{km,t}(Y_t)=0$.

\section{Application to rainfall data}

\subsection{Data}

We use rainfall data from twelve meteorological stations across Germany (see Figure \ref{map}). \\

\begin{figure}
%\flushleft
\fbox{\includegraphics[scale=0.3]{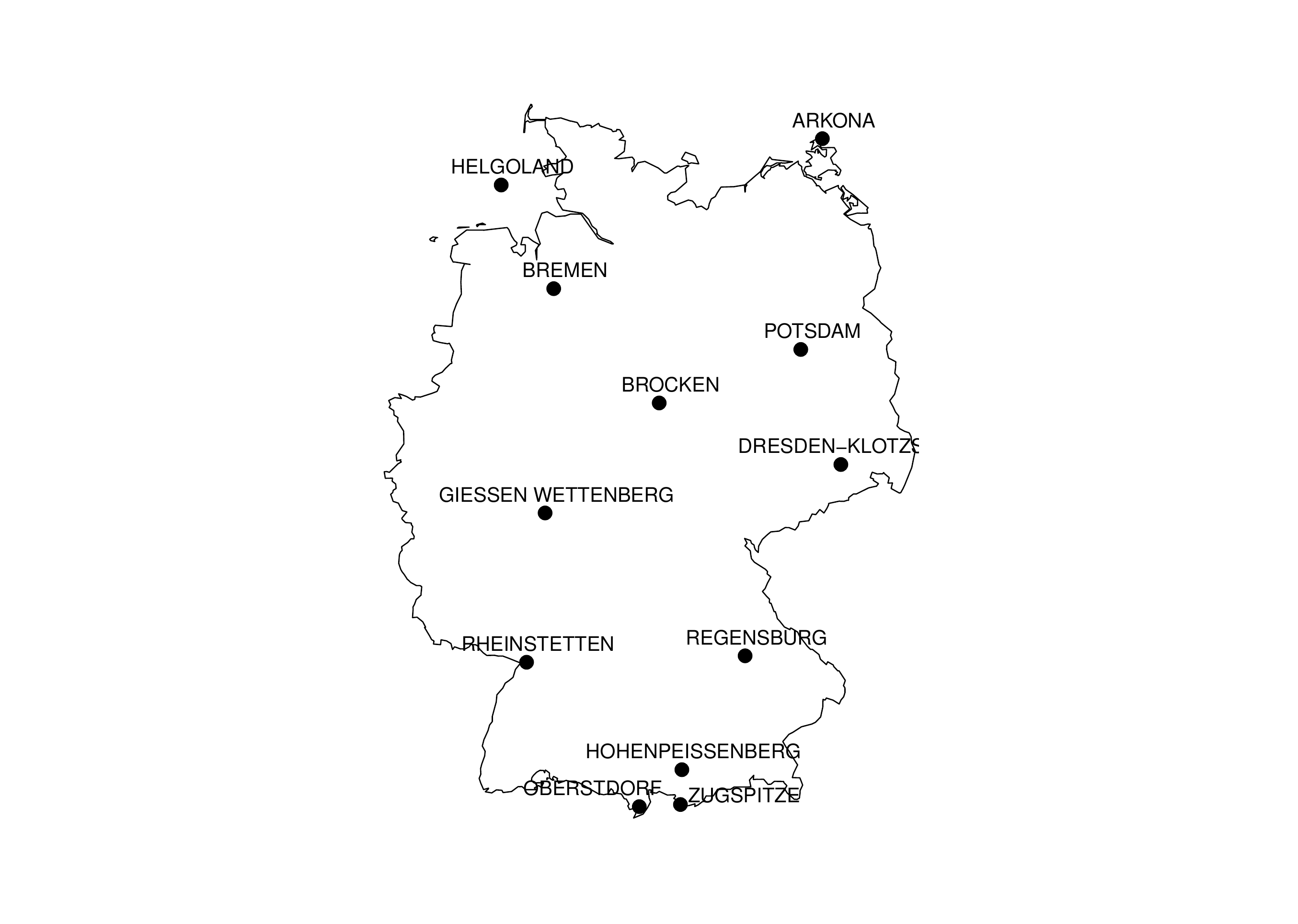}}
\caption{Locations of the stations in Germany}\label{map}
\end{figure}

The source of the data is the \emph{European Climate Assessment and Dataset} (ECA\&D project: \url{http://www.ecad.eu}). For each station, the data consists in daily rainfalls measurements from 01/01/1950 to 12/31/2015. We remove the 16 February 29 so that every year of the period of observation has 365 days. Thus there are 24090 data points left. Missing data are replaced by drawing at random a value among those corresponding to the same day of the year.\\

Table \ref{precip_moy_ann_sta} contains basic statistics for these stations. For most of these stations, the average annual rainfall lies around $600/700$ millimeters, with an average rainfall on rainy days around $3$ to $4$ millimeters. However, four stations (Hohenpeissenberg, Zugspitze, Brocken and Oberstdorf) have much higher precipitations. These stations are located in mountains, which explains why their climate is different from that of the other stations.\\

\begin{table}
\centering
\begin{tabular}{|cccc|}
  \hline
 Station & \parbox[t]{3cm}{\centering Average annual \\rainfall. (mm)} & \parbox[t]{3cm}{\centering Proportion of \\rainy days} & \parbox[t]{3cm}{\centering Average rainfall of rainy days (mm)}\\
  \hline
BREMEN & 699.26 & 0.53 & 3.60 \\ 
  HOHENPEISSENBERG & \textbf{1175.09} & 0.52 & \textbf{6.18} \\ 
  POTSDAM & 590.73 & 0.47 & 3.41 \\ 
  ZUGSPITZE & \textbf{2013.49} & \textbf{0.60} & \textbf{9.18} \\ 
  HELGOLAND & 734.73 & 0.53 & 3.78 \\ 
  DRESDEN-KLOTZSCHE & 661.91 & 0.49 & 3.70 \\ 
  BROCKEN & \textbf{1729.27} & \textbf{0.72} & \textbf{6.56} \\ 
  RHEINSTETTEN & 809.45 & 0.46 & 4.84 \\ 
  GIESSEN WETTENBERG & 635.04 & 0.48 & 3.62 \\ 
  ARKONA & 548.08 & 0.46 & 3.24 \\ 
  OBERSTDORF & \textbf{1765.25} & 0.55 & \textbf{8.84} \\ 
  REGENSBURG & 650.08 & 0.47 & 3.75 \\    \hline
\end{tabular}
\caption{Mean rainfalls by station}\label{precip_moy_ann_sta}
\end{table}

Although the observation space of the model we introduced in Section \ref{model} is continuous, real world observations are not because the precision of the measurements is finite. In our case, this precision is $0.1$mm.

\subsection{Discretization}

The discreteness of the space of observations would not be much of a problem if the order of magnitude of the precipitations were higher than that of the precision. The continuous model would remain a good approximation. However, this is not the case because a good proportion of the precipitation values are comparable to the precision. For example, around $63\%$ (this figure can vary according to the stations) of the measurements are lower than $1$mm. Therefore, it seems necessary to account for the discreteness of the data.\\

We will consider that we do not observe $Y_t$ but
$$\tilde{Y}_t := 0.1\left\lfloor 10Y_t\right\rfloor,$$
where $\left\lfloor \cdot \right\rfloor$ refers to the floor function. Hence, the density of $\tilde{Y}_t$ (with respect to the counting measure on $\{0.1j, j\geq 0\}$) in state $(k,m)$ is given by

$$f_{km,t}(\tilde{y}) =\left\lbrace\begin{array}{lr}
\mathbf{1}_{\tilde{y} = 0}, & m = 1\\
\left(1 - \exp\left[-\frac{0.1\lambda_{km}}{s_k(t)}\right]\right)\exp\left(-\frac{\lambda_{km}\tilde{y}}{s_k(t)}\right)\mathbf{1}_{\tilde{y}>0}, & m\geq 2
\end{array}\right..$$

In other words,

$$10\tilde{Y}_t\mid (X_t^{(1)}=k)\sim p_{k1}\delta_0+\sum_{m=2}^Mp_{km}\mathcal{G}\left(1-\exp\left(-\frac{0.1\lambda_{km}}{s_k(t)}\right)\right),$$
where $\mathcal{G}(\alpha)$ is the geometric distribution with parameter $\alpha$, defined, for $k\in\mathbb{N}$, by $\mathbb{P}(\{k\})=\alpha (1-\alpha)^k$. The consistency result given in Section \ref{model} still holds, the proofs follow the same lines.\\

The results presented in the rest of this paper relate to the discretized model.

\subsection{Results}

The following results correspond to the station of Bremen and the hyperparameters $K=4$, $M=3$ and $d=2$, meaning that there are four hidden states, a mixture of a Dirac mass and two exponential distributions, and seasonalities are trigonometric polynomials with degree $2$. The problem of the choice of $K$, $M$ and $d$ will not be adressed in this paper, although it will be discussed briefly in our conclusion. In this case, we chose $K$, $M$ and $d$ by trying several combinations. The period of the seasonality being one year, $T$ is set to $365$. The parameters of the model have been estimated using the EM algorithm described in paragraph \ref{fit}. We ran it using $40$ different randomly chosen initializations.\\

The (rounded to $10^{-2}$) estimated parameters for the transition matrix $\mathbf{Q}$, the weights matrix $\mathbf{p}$ and the parameters of the exponentials $\boldsymbol{\Lambda}$ are:

$$\hat{\mathbf{Q}} = \begin{bmatrix}
0.71 & 0.12 & 0.13 & 0.04 \\ 
  0.01 & 0.40 & 0.42 & 0.17 \\ 
  0.20 & 0.20 & 0.46 & 0.15 \\ 
  <0.01 & 0.23 & 0.15 & 0.62 \\ 
\end{bmatrix}$$

$$\hat{\mathbf{p}} = \begin{bmatrix}
0.96 & 0.00 & 0.04 \\ 
  <0.01 & 0.19 & 0.81 \\ 
  0.42 & 0.20 & 0.38 \\ 
  <0.01 & 0.19 & 0.81 \\
\end{bmatrix},\quad\hat{\boldsymbol{\Lambda}}=\begin{bmatrix}
0.20 & 0.20 \\ 
2.30 & 0.41 \\ 
2.21 & 13.65 \\ 
0.19 & 0.18 \\ 

\end{bmatrix}$$
The stationary distribution associated with $\hat{\mathbf{Q}}$ is $\hat{\pi}=\begin{bmatrix}
0.21 & 0.24 & 0.3 & 0.24
\end{bmatrix}$, which shows that the four states are well distributed. The corresponding transition diagram is presented in Figure \ref{trans-diag}. The first column of the matrix $\hat{\mathbf{p}}$ corresponds to the probabilities of observing a dry day in each state.  The coefficients of the matrix $\hat{\boldsymbol{\Lambda}}$ give insight into the precipitation intensity in each state. The seasonalities $s_k(t)$ for each state $k$ and $1\leq t\leq 365$ are represented in Figure \ref{saisons}. By examining the estimated parameters, one can give meaning to the states. For example, state $1$ is mostly dry, but when it is not, the rainfalls are heavy in summer, and light in winter. On the other hand, state $4$ is mostly rainy with, on average, heavy rainfalls. The seasonality in this state has a low amplitude.\\

\begin{figure}
\centering
\begin{tikzpicture}[thick,scale=0.6, every node/.style={transform shape}]
\node[state, line width = 0.5mm] (1) at (0,0) {State 1};
\node[state, line width = 0.5mm] (2) at (6,0) {State 2};
\node[state, line width = 0.5mm] (4) at (0,-6) {State 4};
\node[state, line width = 0.5mm] (3) at (6,-6) {State 3};
 \draw[every loop, => latex, minimum width = 0mm]
 			(1) edge[loop above, line width = 0.71mm] node {0.71} (1)
			    edge[bend right=15, auto = left,line width = 0.12mm] node {0.12} (2) 			
			    edge[bend right=15, auto = right,line width = 0.04mm] node {0.04} (4)
			    edge[bend right = 12, auto = right, line width = 0.13mm] node {0.13} (3)
 			(2) edge[loop above, line width = 0.40mm] node {0.40} (2)
                edge[bend right=15, auto = left, line width = 0.42mm] node {0.42} (3)
                edge[bend right=15, auto = right, line width = 0.01mm] node {0.01} (1)
			    edge[ auto = right, bend right = 30, line width = 0.17mm] node {0.17} (4)
            (3) edge[bend right=15, auto = right,line width = 0.2mm] node {0.20} (2)
                edge[loop below, line width = 0.46mm] node {0.46} (3)            
                edge[bend right = 15, auto = left, line width = 0.15mm] node {0.15} (4)            
			    edge[bend right = 12, auto = right, line width = 0.2mm] node {0.20} (1)
            (4) edge[loop below, line width = 0.62mm] node {0.62} (4)			
			    edge[bend right = 15, auto = right, line width = 0.15mm] node {0.15} (3)
			    edge[bend right=15, auto = left, dashed] node {$\simeq 0$} (1)
			    edge[ auto = right, bend right = 30, line width = 0.23mm] node {0.23} (2);
\end{tikzpicture}
\caption{Estimated state transition diagram}\label{trans-diag}
\end{figure}
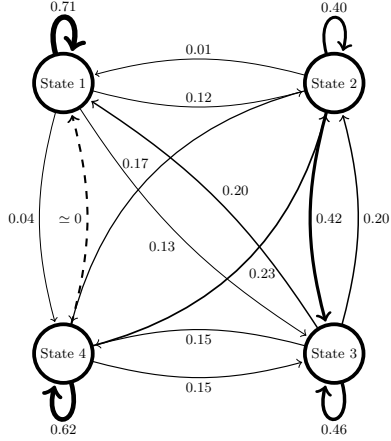

\begin{figure}
\centering
\includegraphics[scale=0.8]{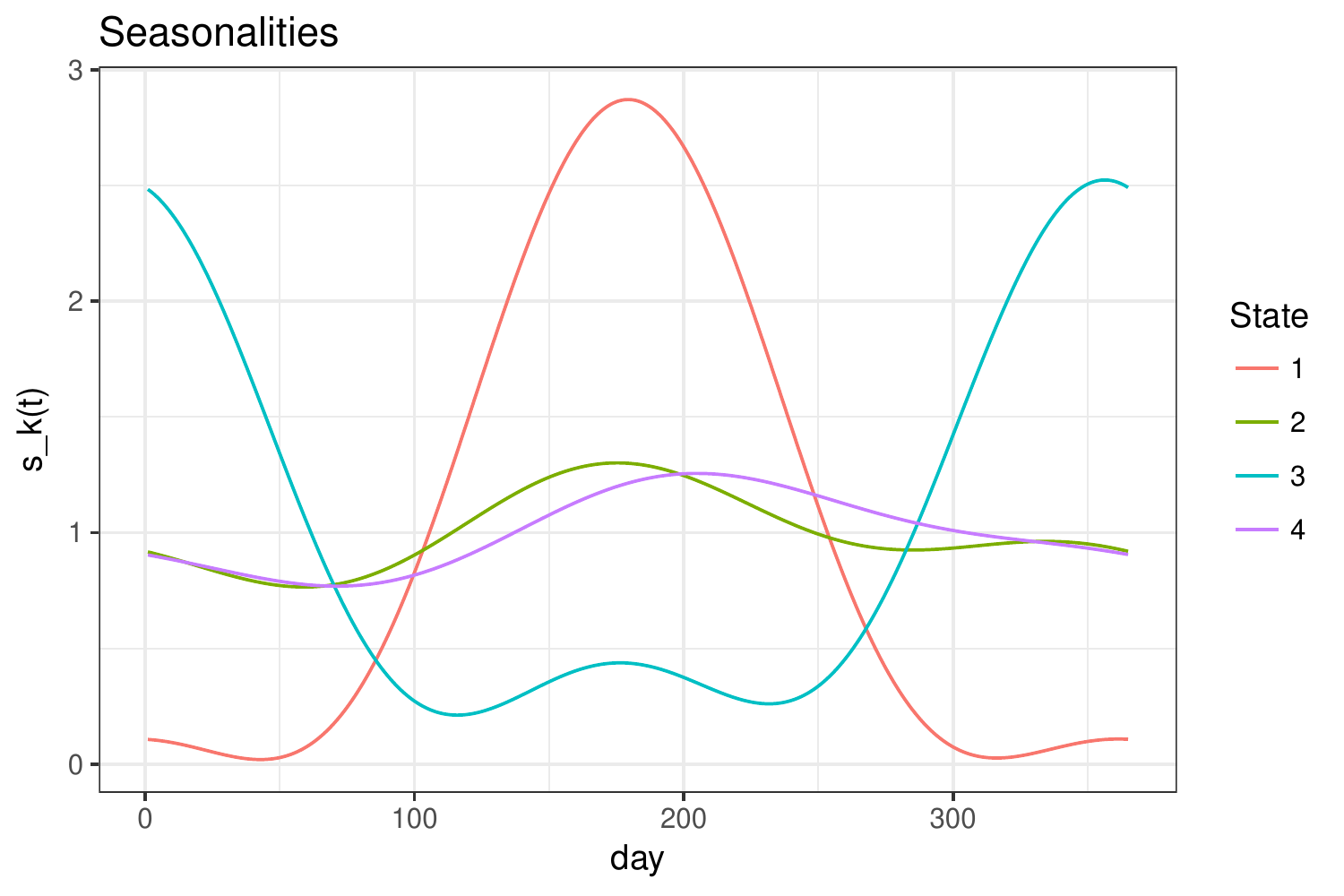}
\caption{Seasonalities $s_k(t)$}\label{saisons}
\end{figure}

Once the estimation of the parameters is done, it can be useful to compute the most probable state at each time step using the forward-backward algorithm. As it outputs the smoothing distributions, we can just use the \emph{maximum a posteriori} rule.
%$$\hat{X}_t^{(1)}:=\argmax_{1\leq k\leq K}\pi_{t\mid n}(k).$$
However, this sequence of states is not necessarily the most likely one. The most likely sequence of states is given by the Viterbi algorithm \citep{viterbi1967}. 
Besides the analysis of the parameters, this can provide another way to give an interpretation to the states. Yet, one must keep in mind that, whether using the maximum a posteriori or Viterbi, the estimated states can be different from the real states. Simulations show that this is particularly the case when the emission distributions are not well separated.

\subsection{Validation of the model}

Recall that our goal is to produce realistic time series of precipitations for a given site. Thus, in order to validate the model, we simulate according to the model, using the estimated parameters, and compare those simulations to the observed time series. To be specific, $1000$ independent simulations are produced, each of them having the same length as the observed series. To perform a simulation, we first simulate a Markov chain $\left(X_t^\mathrm{sim}\right)_{t\geq 1}$ with initial distribution $\hat{\pi}$ and transition matrix~$\hat{\mathbf{Q}}$. Then we simulate the observation process $\left(Y_t^\mathrm{sim}\right)_{t\geq 1}$ using the estimated densities $f_{X_t^\mathrm{sim},t}^{\hat{\theta}_Y}$.\\

Several criteria are considered to carry out the comparison: daily statistics (moments, quantiles, maxima, rainfall occurrence), overall distribution of precipitations, distribution of annual maximum, interannual variability, distribution of the length of dry and wet spells. Each of these statistics is computed from the simulations, which provides an approximation of the distribution of the quantity of interest under the law of the generator (in other words, we use \emph{parametric bootstrap}), hence a $95\%$ prediction interval.\\

Let us first compare the overall distributions of simulated and observed precipitations. Figure \ref{qqplot} depicts the quantile-quantile plot of these two distributions. The match is correct, except in the upper tail of the distribution. The last point corresponds to the maximum of the simulated values, which is much larger than the maximum observed value. This should not be considered as a problem: a good weather generator should be able to (sometimes) generate values that are larger than those observed.\\

\begin{figure}
\centering
\includegraphics[scale=0.5]{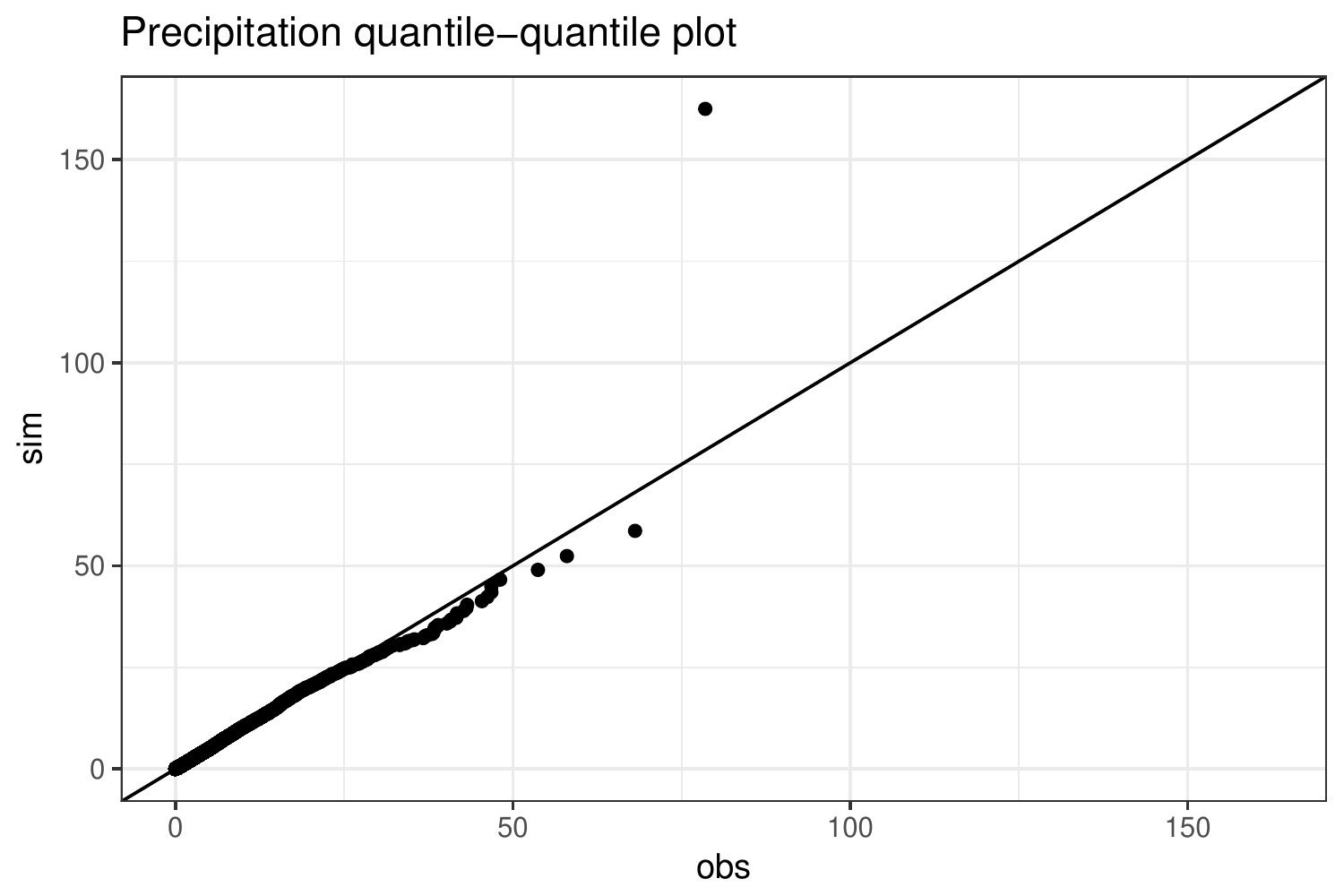}
\caption{Quantile-quantile plot of observed versus simulated precipitations distributions}\label{qqplot}
\end{figure}

We then focus on daily distributions. Figure \ref{moments} shows the results obtained for the first four daily moments and for the daily frequency of rainfall. It shows that these statistics are well reproduced by the model. Even though we did not introduce seasonal coefficients in the weights of the mixture densities, a seasonality appears in the simulated occurrence process. This can be explained by the combined effects of discretization and seasonality in the intensity of rainfall (i.e. $s_k(t)$). For example, Figure \ref{saisons} shows that state $1$ will produce lots of small values in winter. All the values smaller than $0.1$ will be set to $0$ by discretization, hence a higher proportion of dry days in state $1$ in winter. The same phenomenon is observable in state $3$ in summer.\\ 

\begin{figure}
\centering
\includegraphics[scale=0.5]{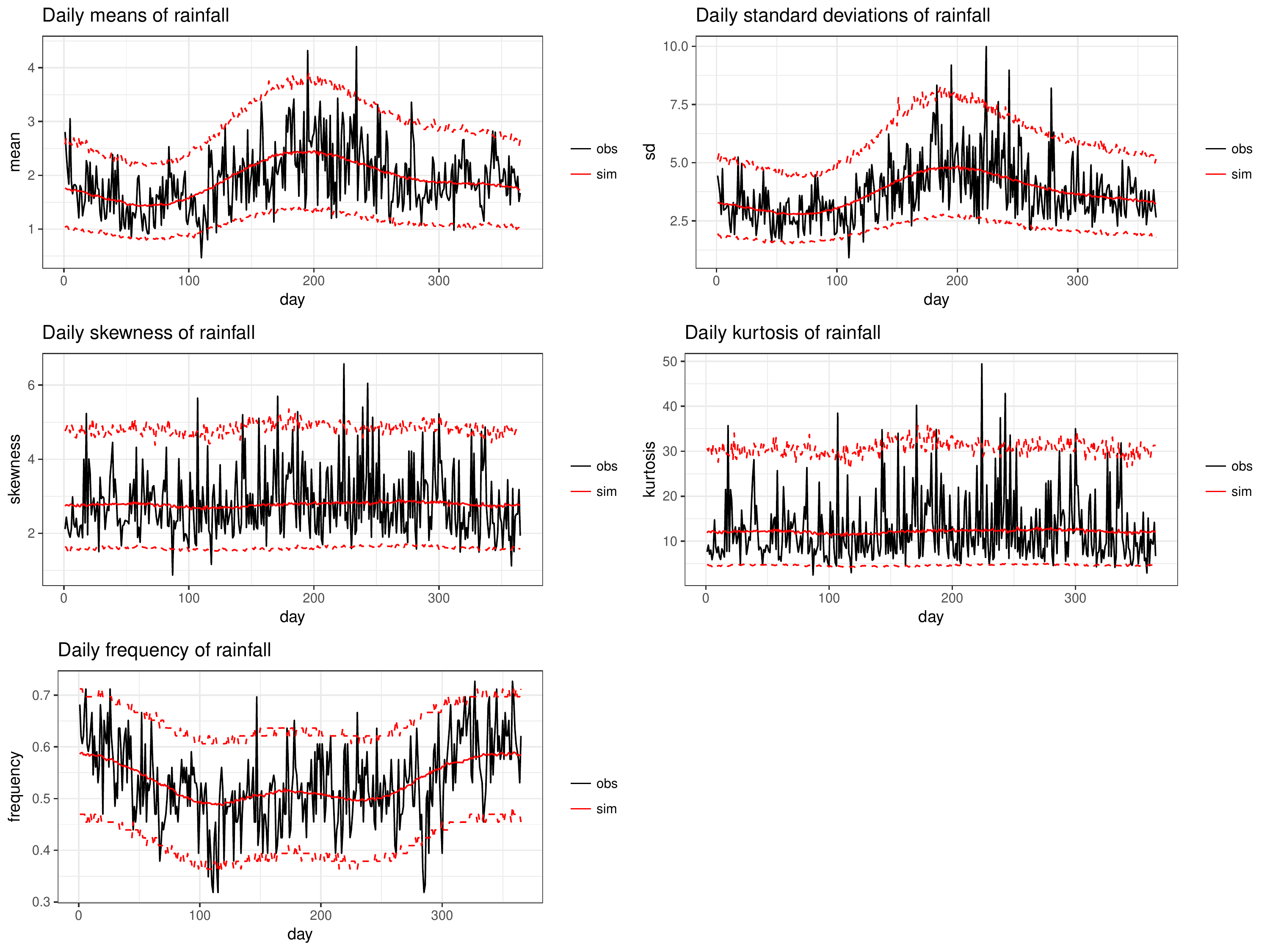}
\caption{Daily moments and frequency of precipitations. The black line relates to observations, the red solid line is the mean over all simulations, and the dashed lines depict an estimated $95\%$ prediction interval under the model.}\label{moments}
\end{figure}

We can also compute, according to the same principles, the daily quantiles (see Figure \ref{quantiles}). Once again, the match is correct, even for the highest quantiles. Note that we did note consider lower quantiles as they would be $0$ because of the high proportion of dry days.\\

\begin{figure}
\centering
\includegraphics[scale=0.5]{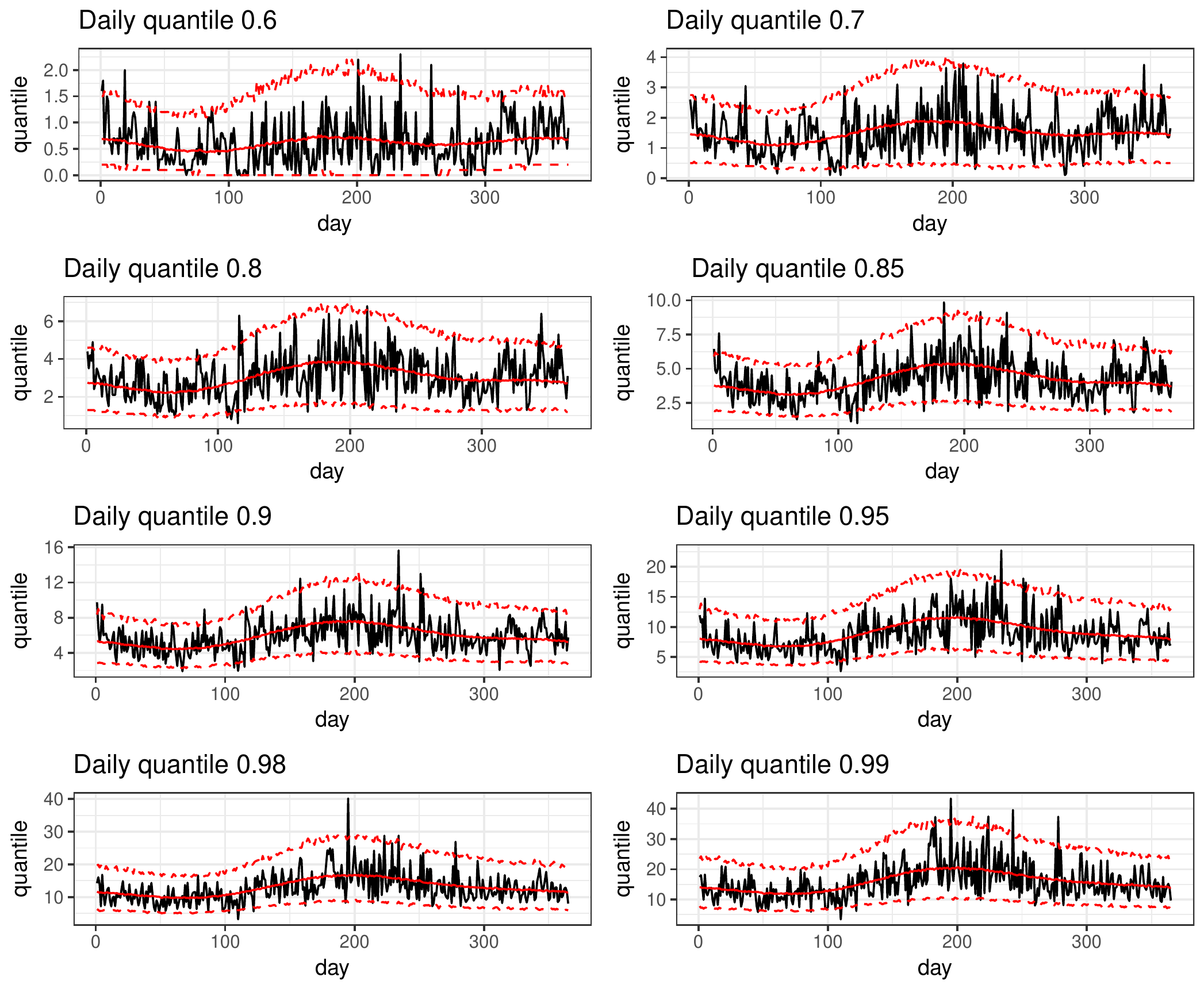}
\caption{Daily quantiles of precipitations. The black line relates to the observations, the red solid line is the mean over all simulations, and the dashed lines depict an estimated $95\%$ prediction interval under the model.}\label{quantiles}
\end{figure}

We also investigated daily maxima and annual maxima. By daily maxima we mean that for each day of the year, we consider the record of precipitation for that day during the observation period (for example, $22.1$mm for January 1st). Regarding daily maxima, we are interested in the distribution of the yearly maximum of precipitations. The results are presented in Figure \ref{max}. Note that the intervals materialized by the dashed lines are just intervals within which $95\%$ of the simulated values can be found. The maximum simulated values are larger than the observed ones.\\

\begin{figure}
\centering
\includegraphics[scale=0.5]{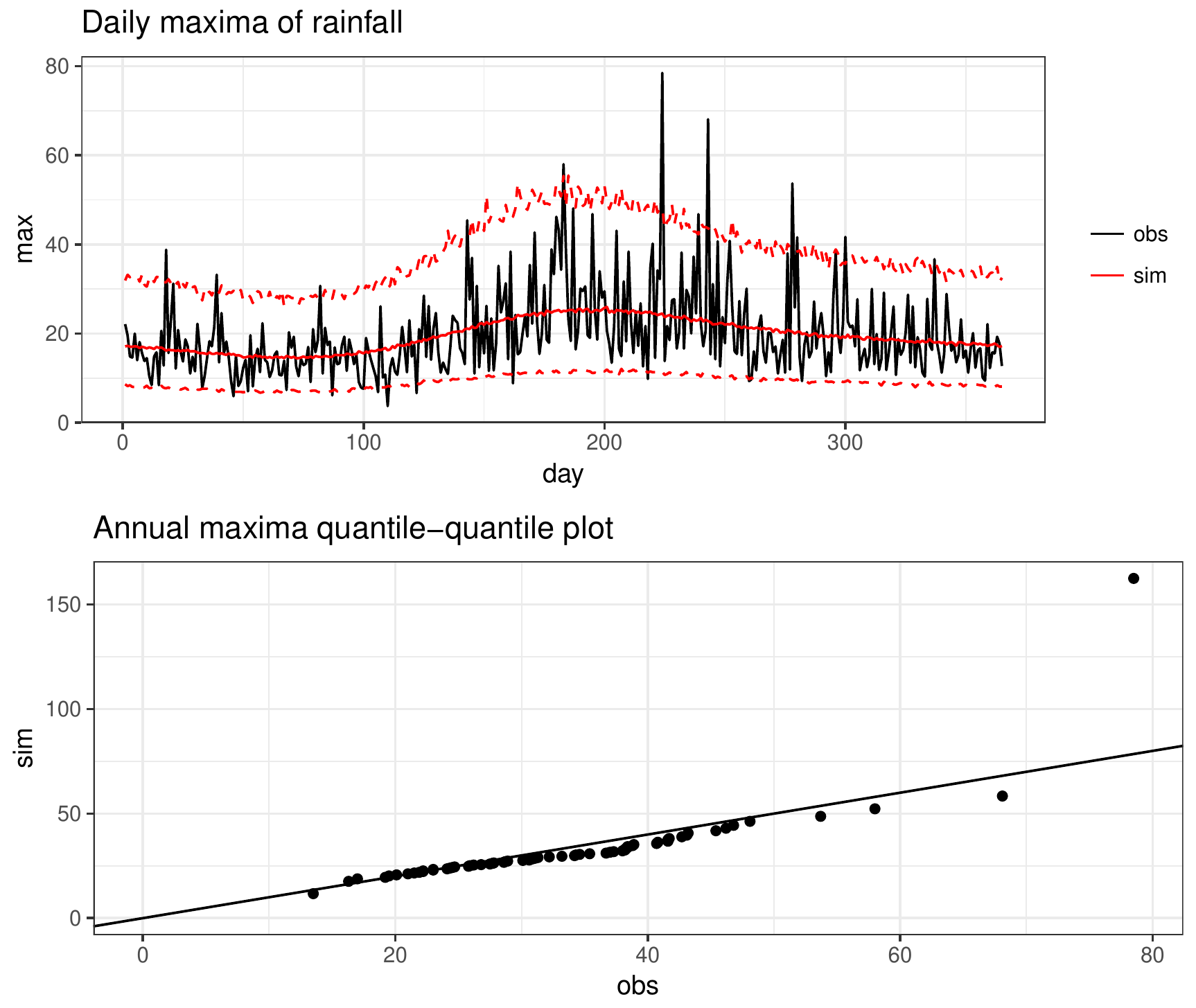}
\caption{Observed versus simulated daily maxima (top graph) and yearly maxima (bottom graph)}\label{max}
\end{figure}

The distribution of the duration of \emph{dry and wet spells} is another quantity of interest that we studied. A wet (resp. dry) spell is a set of consecutive rainy (resp. dry) days. This statistic provides a way to measure the time dependence of the occurrence process. The results are presented in Figure \ref{spells}. The dry spells are well modelled, whereas there is a slight underestimation of the frequency of 2-day wet spells while the single day events frequency is slightly overestimated.\\

\begin{figure}
\centering
\includegraphics[scale=0.5]{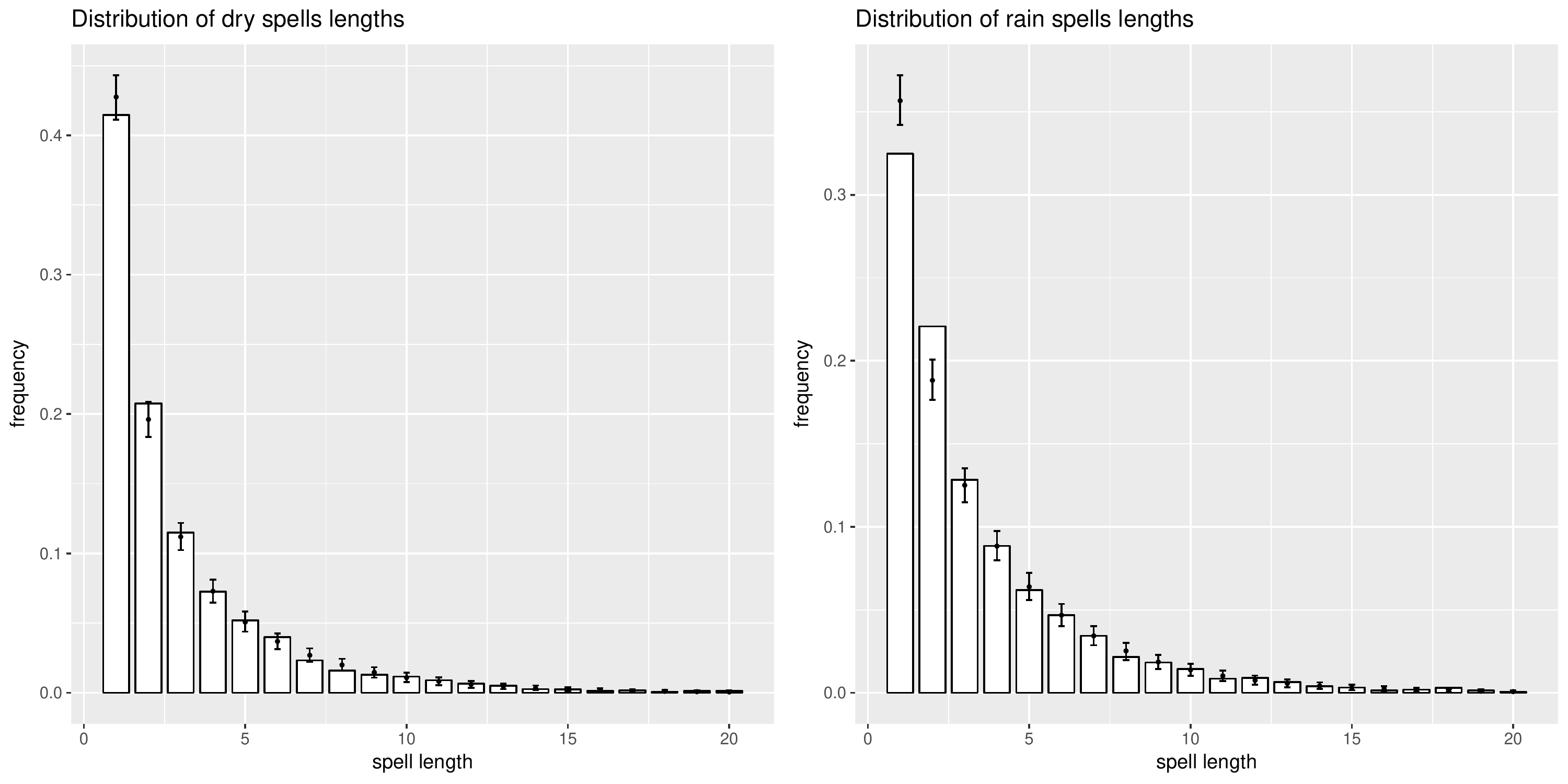}
\caption{Distribution of the lengths of dry (left plot) and wet (right plot) spells: observed (bars) versus simulated (error bars). The dots represent the means of the simulations}\label{spells}
\end{figure}

Stochastic precipitations generators often underestimate the \textit{interannual variability} of precipitations \citep{Katz1998}. We focus on the total yearly/monthly  total precipitations and we look at their interannual variability. The histogram in Figure \ref{annual} is the observed distribution of yearly total precipitations (that is $66$ observations). The line is a kernel density estimation of simulated yearly precipitations.  

\begin{figure}
\centering
\includegraphics[scale=0.3]{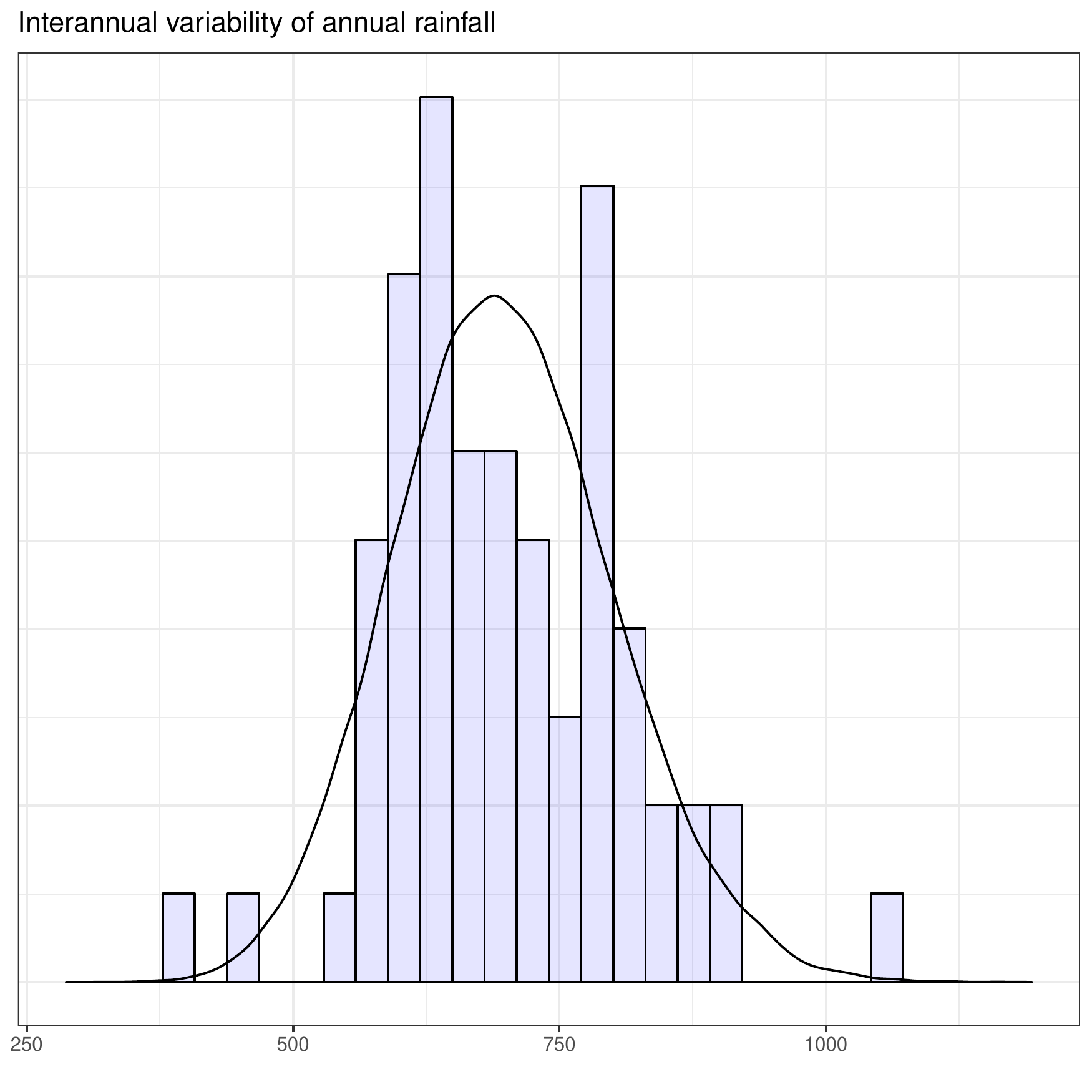}
\caption{Interannual variability of total yearly precipitations: observed (histogram) and simulated (line)}\label{annual}
\end{figure}

We have performed the same computations on monthly precipitations (see Figure \ref{month}). Our model does not underestimate interannual variability.

\begin{figure}
\centering
\includegraphics[scale=0.5]{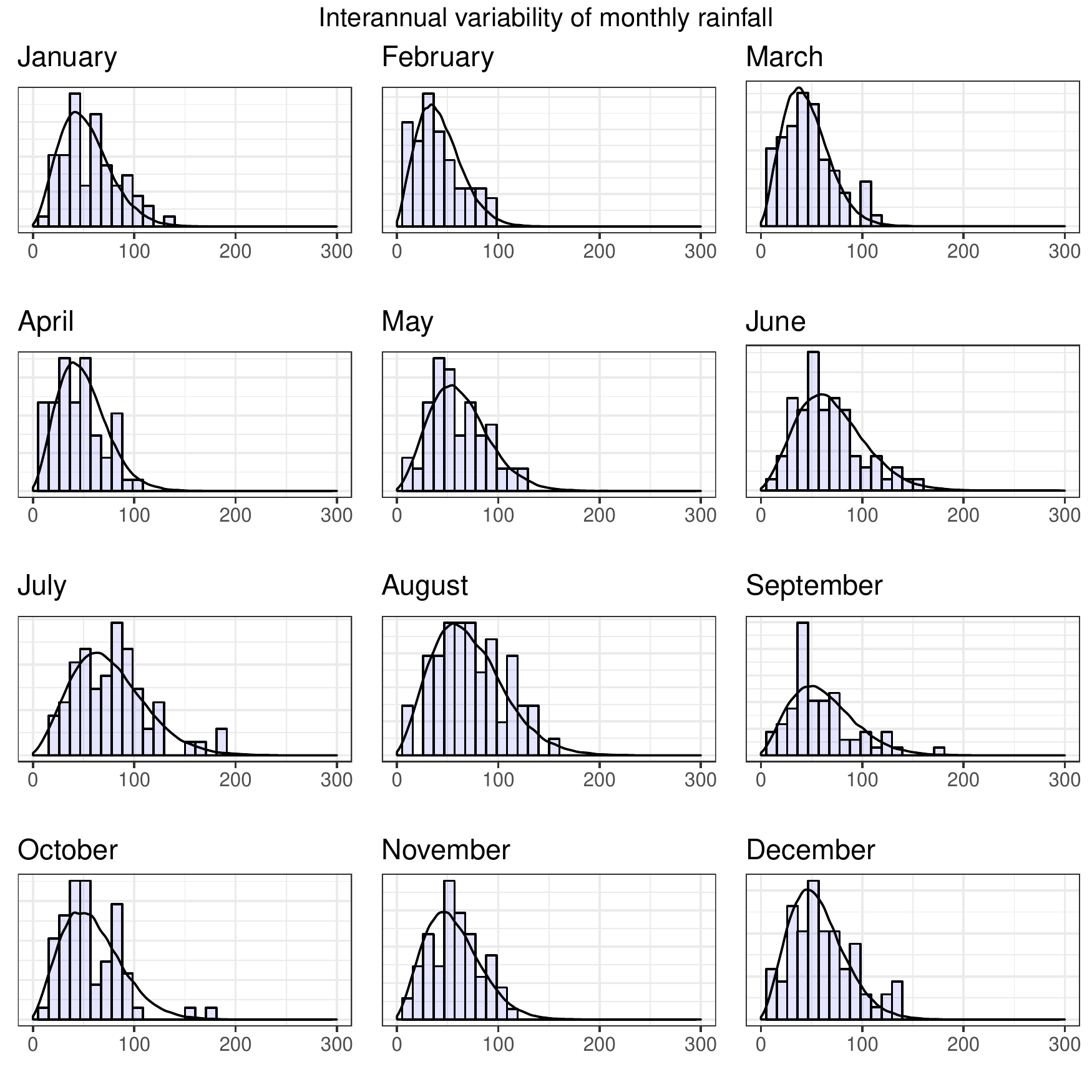}
\caption{Interannual variability of total monthly precipitations: observed (histogram) and simulated (line)}\label{month}
\end{figure}

\paragraph{}
This validation procedure was used on all of the twelve studied stations. For eight of them, the validation results are comparable to those of Bremen. However, four stations give less satisfying results: Hohenpeissenberg, Oberstdorf, Zugspitze and Brocken. The first three are located in the same mountain area of Bavaria (Zugspitze being the highest summit in Germany). The station of Brocken is located at the summit of a mountain (highest peak in Northern Germany) and is well-known for its extreme climatic conditions (see \url{http://www.brocken.climatemps.com/}) and its subarctic climate (i.e. cold winters, no dry season and short, cool summers). Thus, the climate of these stations strongly differs from the others. This suggests that another model should be used for these stations.

\paragraph{Impact of discretization}

We compared the previous results with those obtained when no discretization of the model is performed. The most notable difference lies in the third state, which corresponds to a low intensity of precipitations. The proportion of dry days rises from $0.42$ to $0.7$ and the shape of the seasonality is affected. One also observe some changes in the transition matrix. Validation results are also better with the discretized model.

\section{Conclusion and future work}

%Conclusion + gamma/Pareto distributions + seasonal $p_k$ + choosing $K,m,d$ + autres stations.

\paragraph{}In this paper, we introduced a new stochastic weather generator for daily rainfall. It consists in a seasonal version of a hidden Markov model. The seasonality of precipitation time series is modelled by using periodic coefficients in the emission distributions. Those parameters, along with the transition matrix, were estimated through maximum likelihood inference. Using methods developed for hidden Markov models, we proved that this model is identifiable and that under reasonable assumptions, the maximum likelihood estimator is strongly consistent. The performance of the model has been assessed by comparing simulations to the observations, our goal being the simulation of realistic time series. For the considered criteria, the model gave good results, which proves that the simulated times series share many statistical properties with the observed data. By analyzing the estimated parameters, we also gave physical interpretations to the different states.

\paragraph{}This model could be explored and developed in several ways. 

\begin{itemize}
\item Other choices of emission distributions could be made. We checked that replacing the mixtures of two exponential distributions by a single gamma distribution (in each state) also leads to good results. Contrary to (mixtures of) exponential distributions, gamma distributions have the advantage of being able to have a positive mode, which can be useful in some cases. Mixtures of gamma distributions can also be considered to increase flexibility. On the other hand it decreases parcimony. However, these are light tail distributions. When interested in extreme values of precipitations, it may be necessary to introduce a heavy tail distribution such as the generalized Pareto distribution \citep{lennartsson2008}.

\item Although we did not model directly the seasonality of the occurrence of precipitations, we obtained it as a by-product of discretization. Yet it is possible to introduce seasonality in the weights of the mixtures, for example by using the logistic function applied to a trigonometric polynomial depending on the state. Although it increases the complexity of the model, it is more realistic, as the seasonality in the occurrence of precipitations clearly appears when looking at the data. We actually fitted such a model to the same data, using gamma emission distributions. The validation procedure also leads to good results, yet not much better than those of the original model.

\item In this model, the seasonal variations of the distribution of precipitations are included in the emission distributions, the underlying Markov chain being homogeneous. It is also possible to consider a non-homogeneous hidden Markov chain by replacing the constant transition matrix $\mathbf{Q}$ by a periodic function of time $\mathbf{Q}(t)$, at the cost of an increased number of parameters.

\item We proved a consistency result for cyclo-stationary processes: it does not apply to other sources of non-stationarity, such as long term trends that can be observed on many meteorological time series. Here we had checked first that it was not necessary to include a trend. However, when studying other climates or other variables, it can be necessary.
\end{itemize}

\paragraph{}We did not adress the issue of \emph{model selection}. We have three hyper-parameters to choose: the number of states $K$, the complexity of the emission distributions $M$, and the degree of the seasonality $d$. The number of parameters of the model is quadratic in $K$ and linear in $M$. Thus it might be intractable if we do not choose the hyper-parameters with parcimony, because of an unrealistic computation time and a large number of local maxima in the likelihood function. Furthermore, from an applied point of view, a large number of states often means a loss in interpretability. On the other hand, a too small number of parameters does not allow a good fit to the data. The problem of choosing the number of states when fitting a HMM is challenging. Although not fully justified in theory in the framework of HMM, the \emph{Bayesian Information Criterion} \citep{schwarz1978} is very popular for this purpose. The BIC is a penalized likelihood criterion, which realizes a trade-off between fitting the data and being parcimonious. Another approach is cross-validated likelihood \citep{celeux2008}, even though it is computationnally intensive.  In \cite{lehericy2016order}, the author introduces a penalized least square estimator for the order of a nonparametric HMM and proves its consistency. He also estimates the number of states by thresholding the spectrum of the empirical version of the matrix $\mathbf{N}$ in the spectral algorithm (see equation \eqref{eq3}) and shows that this procedure leads to a consistent estimator of $K$. In addition, he proves results on consistent estimation of $K$ using penalized maximum likelihood with a BIC-like penalty \citep{lehericy2017}. However, when dealing with real world data, other considerations should be taken into account, such as interpretability of the states, computing time, or the ability of the model to reproduce some behaviour of the data, as explained in \cite{bellone2000}. Indeed, according to \cite{pohle2017}, the popular AIC (\emph{Akaike Information Criterion}) and BIC, as well as other penalized criteria, tend to overestimate the number of states as soon as the data generating process differs from a HMM (e.g. the presence of a conditional dependence). Hence it is advised to use such a criterion as a guide, without following it blindly. 

\section*{Acknowledgements}
The author would like to thank Yohann De Castro, \'Elisabeth Gassiat, Sylvain Le Corff and Luc Leh\'ericy from Universit\'e Paris-Sud for fruitful discussions and valuable suggestions. This work is supported by EDF. We are grateful to Thi-Thu-Huong Hoang and Sylvie Parey from EDF R\&D for providing this subject and for their useful advice. 

\section*{Supplementary material}

The dataset used in the paper is available at \url{https://www.math.u-psud.fr/~touron/data}.

\bibliographystyle{imsart-nameyear}
\bibliography{bibli}

\end{document}